\RequirePackage{fix-cm}
\documentclass[smallextended,numbook]{svjour3}       
\smartqed  
\usepackage{graphicx}
\usepackage{amsmath,amsfonts,amssymb,dsfont,bbm,bm, hyperref,graphicx}
\usepackage[english]{babel}

\DeclareMathOperator{\Tr}{Tr}
\newcommand\bq{\begin{eqnarray}}
\newcommand\eq{\end{eqnarray}}
\newcommand\nn{\nonumber}
\renewcommand{\epsilon}{\varepsilon}
\newcommand\cF {\mathcal{F}}
\newcommand\cH{\mathcal{H}}
\newcommand\cD {\mathcal{D}}
\newcommand\cE {\mathcal{E}}
\newcommand\cM {\mathcal{M}}
\newcommand\R {\mathbb{R}}

\journalname{ARMA}
\begin{document}

\title{The Bogoliubov free energy functional I}
\subtitle{Existence of minimizers and phase diagram}


\author{Marcin Napi\'orkowski \and Robin Reuvers \and Jan Philip Solovej}


\institute{M. Napi\'orkowski \at
              Institute of Science and Technology Austria, Am Campus 1, 3400 Klosterneuburg, Austria        
             \at Department of Mathematical Methods in Physics, Faculty of Physics, University of Warsaw, Pasteura 5, 02-093 Warsaw, Poland  \\
		\email{marcin.napiorkowski@fuw.edu.pl}   
           \and
           R. Reuvers \at
              QMATH, Department of Mathematical Sciences, University of Copenhagen, Universitetsparken 5, DK-2100 Copenhagen \O, Denmark \\
             \emph{Present address:} DAMTP, Centre for Mathematical Sciences, University of Cambridge, Wilberforce Road, Cambridge CB3 0WA, United Kingdom \\
            \email{r.reuvers@damtp.cam.ac.uk}              
           \and
           J.P. Solovej \at
              QMATH, Department of Mathematical Sciences, University of Copenhagen, Universitetsparken 5, DK-2100 Copenhagen \O, Denmark \\
              \email{solovej@math.ku.dk}     
}

\date{\phantom{.}} 

\maketitle

\begin{abstract}
  The Bogoliubov free energy functional is analysed. The functional
  serves as a model of a translation-invariant Bose gas at positive
  temperature. We prove the existence of minimizers in the case of
  repulsive interactions given by a sufficiently regular two-body
  potential. Furthermore, we prove existence of a phase transition in
  this model and provide its phase diagram.
\end{abstract}

\tableofcontents

\section{Introduction}
Almost all work in the field of interacting Bose gases has its genesis
in Bogoliubov's seminal 1947 paper \cite{Bogoliubov-47b}. In this
work, Bogoliubov proposed an approximate theory of interacting bosons
in an attempt to explain the superfluid properties of liquid
Helium. Since then, his model has widely been used to study bosonic
many-body systems, particularly in the 1950s and 1960s. Despite being
intuitively appealing and undoubtedly correct in many aspects,
Bogoliubov's theory lacked a mathematically rigorous understanding.

The experimental success in achieving Bose--Einstein condensation in
alkali atoms \cite{Ketterle-95} has renewed the interest in the
theoretical description of such systems, and significant
progress was made in the mathematical analysis of Bose gases. We refer
to \cite{LieSeiSolYng-05} for an extensive review. Most of these
results concern the ground state energies of different bosonic
systems.

While Bogoliubov's theory is very useful in relation to these
problems, its primary goal was to determine the excitation spectrum of
a Bose gas. Indeed, the structure of the excitation spectrum derived
by Bogoliubov allowed him to justify Landau's criterion for
superfluidity \cite{Landau-41}, and thus provided a microscopic theory
of this phenomenon. A rigorous justification of Bogoliubov's theory in
that context has been established only recently for a large class
of bosonic systems within the so-called mean-field limit
\cite{Seiringer-11,GreSei-13,LewNamSerSol-15,DerNap-13,NamSei-14} (see
\cite{Sei-14} for a recent review).

Our goal (and that of the accompanying paper
\cite{NapReuSol2-15}) is to give a variational formulation of Bogoliubov's theory
for bosonic systems at positive and zero temperature.
Bogoliubov's original approximation consists in adapting the Hamiltonian so that it is quadratic in creation and
annihilation operators, and we know that ground or Gibbs states of such Hamiltonians are quasi-free (or coherent) states. Here, we reverse the idea, and retain the full Hamiltonian while only varying over Gaussian states (which include the aforementioned classes of states). This gives the variational model that we will study in this paper (see Appendix \ref{app:derivation} for relevant definitions and a derivation).

The hope is that this approach allows us to capture
important physical aspects of the system; the relevant variational states have for example served as trial states in establishing the correct asymptotic bounds on the ground state energy of Bose gases
\cite{Solovej-06,ErdSchYau-08,GiuSei-09}.\\

The Bogoliubov variational theory can be seen as the bosonic
counterpart of Hartree--Fock theory for Fermi gases. More
precisely, it is similar to generalized Hartree--Fock theory, which
includes the Bardeen--Cooper--Schrieffer (BCS) trial states and is
often called Hartree--Fock--Bogoliubov theory. In HFB theory the
trial states are quasi-free states on a
fermionic Fock space (see \cite{BacLieSol-94} for details).

HFB theory is a widely-used tool for understanding fermionic
many-body quantum systems. One of the most prominent examples related
to this approach is the model of superconductivity that is based on the
BCS functional. This model and the related BCS gap equation
have been studied both from the physical and mathematical
point of view (see
e.g. \cite{FraHaiNabSei-07,HaiHamSeiSol-08,HaiSei-08}). 

In this paper, we are interested in the
bosonic counterpart of the BCS functional, or, more precisely, to the
BCS functional with the direct and exchange terms included (as
discussed in \cite{BraHaiSei-14}).\\

Concretely, we want to analyse the model defined by the \textit{Bogoliubov free energy functional} $\cF$ given by
\bq
\label{def:grandcanfreeenergyfunctional}
\begin{aligned}
\mathcal{F}(\gamma,\alpha,\rho_{0})&= (2\pi)^{-3}\int_{\mathbb{R}^3} p^{2}\gamma(p)dp-\mu\rho-TS(\gamma,\alpha)+\frac{\widehat{V}(0)}{2}\rho^{2} \\
&+\frac{1}{2}(2\pi)^{-6}\iint_{\mathbb{R}^3\times \mathbb{R}^3}\widehat{V}(p-q)\left(\alpha(p)\alpha(q)+\gamma(p)\gamma(q)\right)dpdq \\ &+\rho_{0}(2\pi)^{-3}\int_{\mathbb{R}^3}\widehat{V}(p)\left(\gamma(p)+\alpha(p)\right)dp,
\end{aligned}
\eq
which is the free energy expectation value in a quasi-free state (see Appendix \ref{app:derivation} for a derivation). Here, $\rho$ denotes the density of the system and 
\bq
\rho=\rho_0+(2\pi)^{-3}\int_{\mathbb{R}^3}\gamma(p)dp=:\rho_0+\rho_\gamma. \nn
\eq
The entropy $S(\gamma,\alpha)$ is 
\begin{equation}
\label{entro}
\begin{aligned}
&S(\gamma,\alpha)= \quad(2\pi)^{-3}\int_{\mathbb{R}^3}s(\gamma(p),\alpha(p))dp\quad=\quad(2\pi)^{-3}\int_{\mathbb{R}^3}s(\beta(p))dp \\
 &=(2\pi)^{-3}\int_{\mathbb{R}^3}\left[\left(\beta(p)+\frac{1}{2}\right)\ln\left(\beta(p)+\frac{1}{2}\right)-\left(\beta(p)-\frac{1}{2}\right)\ln\left(\beta(p)-\frac{1}{2}\right)\right]dp,
\end{aligned}
\end{equation}
where
\bq
\beta(p):=\sqrt{\left(\frac{1}{2}+\gamma(p)\right)^{2}-\alpha(p)^{2}}.\nn
\eq
The functional is defined on the domain $\mathcal{D}$ given by
\bq
\mathcal{D}=\{(\gamma,\alpha,\rho_{0})|\gamma \in L^{1}((1+p^{2})dp),\gamma(p)\geq0, \alpha(p)^{2}\leq\gamma(p)(1+\gamma(p)), \rho_{0}\geq 0\}. \nn
\eq

This set-up describes the grand canonical free energy of a homogeneous Bose gas at temperature $T\geq0$ and chemical potential $\mu\in \mathbb{R}$ in the thermodynamic limit. The particles interact through a repulsive radial two-body potential $V(x)$. Its Fourier transform is denoted by $\widehat{V}(p)$ and is given by
$$\widehat{V}(p)=\int_{\R^3}e^{-ipx}V(x)dx.$$

The function $\gamma\in L^{1}((1+p^{2})dp)$ describes the momentum
distribution of the particles in the system. Since the total density
equals $\rho=\rho_0+(2\pi)^{-3}\int_{\mathbb{R}^3}\gamma(p)dp$, it
follows that a non-negative $\rho_0$ can be seen as the macroscopic
occupation of the state of momentum zero and is therefore interpreted as
the density of the Bose--Einstein condensate fraction.\footnote{It is mathematically possible to include and study such a condensate at momentum $p\neq0$, but we follow the standard physics approach and assume the condensate, if present, forms at momentum zero.}

Finally, the function $\alpha(p)$ describes pairing in the system and
its non-vanishing value can therefore be interpreted as the macroscopic coherence
related to superfluidity.

It is worthwhile to note that the case $V\equiv0$ describes the well-known non-interacting Bose gas first studied by Einstein. Assuming $T>0$, its explicit minimizer for $\mu\leq0$ is
\[
\gamma(p)=\frac1{e^{(p^2-\mu)/T}-1}
\]
and $\alpha\equiv0$. If $\mu<0$, $\rho_0=0$ (no BEC), and $\rho_0\in[0,\infty)$ for $\mu=0$ (possibly BEC). There is no minimizer for $\mu>0$. Because everything is known explicitly for this case, we will always assume $V\not\equiv0$.

To the best of our knowledge, the functional \eqref{def:grandcanfreeenergyfunctional} appeared for the first
time in the literature in a 1976 paper by Critchley and Solomon
\cite{CriSol-76}. Probably due to its complexity, however, it has never been
analysed - neither from a mathematical nor from a physical point of
view - but simplified versions have been considered. In particular, Angelescu, Verbeure and Zagrebnov \cite{AngVer-95,AngVerZag-97} studied variational models based on modifications of the original Bogoliubov Hamiltonian and these models can be seen as linear approximations to the full functional. They carry less physical information and are considerably easier to treat.  Other simplified models are reviewed in \cite{ZagBru-01}, although not necessarily from a variational perspective. 

Our analysis of the full functional is divided into two parts. In this part, we consider the existence and general properties
of equilibrium states of this model. According to
statistical mechanics, the equilibrium state corresponding to
temperature $T$ and chemical potential $\mu$ is given by the minimizer
of \eqref{def:grandcanfreeenergyfunctional}. The free energy is
therefore
\begin{align}
F(T,\mu)=\inf_{(\gamma,\alpha,\rho_{0})\in \mathcal{D}}\mathcal{F}(\gamma,\alpha,\rho_{0}). \label{def:grandcanonicalminimization}
\end{align}
The physical information about the system at a given $T$ and $\mu$ is
thus encoded in the structure of the minimizers. For example, a
minimizer with $\gamma\equiv0$ and $\rho_0>0$ corresponds to pure
Bose--Einstein Condensation; non-vanishing $\alpha$ signifies the presence of pairing. Hence, any further analysis of the model relies on the
well-posedness of the minimization problem
\eqref{def:grandcanonicalminimization}, which we address first. Knowledge about the minimizers for different $(T,\mu)$  then leads to a phase diagram.
We will also discuss the relation between Bose--Einstein condensation and
superfluidity in translation-invariant systems (see \cite{Bal-04} for
a historical overview on this topic). Our results are stated in the next
section.\\

In the second part of this work \cite{NapReuSol2-15}, we analyse the
functional in the \textit{dilute} (or \textit{low-density})
limit. Although Bogoliubov's primary goal was to provide a description
for liquid helium, which is a strongly-interacting system, it is
generally agreed that his theory is more suitable to describe dilute (hence weakly-interacting) systems. Here,
low density means that the mean interparticle distance $\rho^{-1/3}$
is much larger than the \textit{scattering length} $a$ of the
potential, that is
$$\rho^{1/3}a \ll 1.$$

To be able to analyse the dilute limit, we need to consider the
\textit{canonical} counterpart of
\eqref{def:grandcanfreeenergyfunctional} at fixed density $\rho$ given by \bq
\label{def:canonicalfreeenergyfunctional}
\begin{aligned}
  \mathcal{F}^{\rm{can}}(\gamma,\alpha,\rho_{0})&= (2\pi)^{-3}\int_{\mathbb{R}^3} p^{2}\gamma(p)dp-TS(\gamma,\alpha)+\frac12\widehat{V}(0)\rho^{2}\\ &+\rho_{0}(2\pi)^{-3}\int_{\mathbb{R}^3}\widehat{V}(p)\left(\gamma(p)+\alpha(p)\right)dp \\
  &+\frac{1}{2}(2\pi)^{-6}\iint_{\mathbb{R}^3\times
    \mathbb{R}^3}\widehat{V}(p-q)\left(\alpha(p)\alpha(q)+\gamma(p)\gamma(q)\right)dpdq,
\end{aligned}
\eq
with $\rho_0=\rho-\rho_\gamma$. The canonical minimization problem is
\begin{equation}
 \label{def:canonicalminimization}
\begin{aligned}
F^{\rm{can}}(T,\rho)&=\inf_{\substack{({\gamma},{\alpha},\rho_{0}=\rho-\rho_\gamma)\in\cD\\
    }}\cF^{\rm{can}}(\gamma,\alpha,\rho_0)=\inf_{0\leq\rho_0\leq\rho}f(\rho-\rho_0,\rho_0),
\end{aligned}
\end{equation}
where
\begin{equation}
\label{funcf}
f(\lambda,\rho_0)=\inf_{\substack{({\gamma},{\alpha})\in\cD'\\
      \int\gamma=\lambda
    }}\cF^{\rm{can}}(\gamma,\alpha,\rho_0)
\end{equation}
and 
\[
\mathcal{D}'=\{(\gamma,\alpha)\ |\ \gamma\in L^1((1+p^2)dp),\ \gamma(p)\geq0,\ \alpha(p)^2\leq \gamma(p)(\gamma(p)+1)\}.
\]
Strictly
speaking, this is not really a canonical formulation: it is only the
expectation value of the number of particles that we fix. We will
nevertheless describe this energy as canonical. The function
$F(T,\mu)$ as a function of $\mu$ is the Legendre transform of the
function $F^{\rm{can}}(T,\rho)$ as a function of $\rho$.

Having given a proper meaning to the notion of diluteness, one can now
ask different questions regarding the low-density limit. One
particularly interesting problem is how interactions influence the critical temperature (i.e.\ the temperature of the
phase transition between the condensed and non-condensed phase) in a
weakly-interacting Bose gas.  It is nowadays agreed that the
transition temperature should change linearly in $a$, that is,
\bq
\frac{\Delta T_{\rm{c}}}{T_{\rm{fc}}} \approx c \rho^{1/3}a \nn 
\eq
with $c>0$. Here $\Delta T_{\rm{c}}=T_{\rm{c}}-T_{\rm{fc}}$, where
$T_{\rm{c}}$ is the critical temperature in the interacting model
and $T_{\rm{fc}}=c_0\rho^{2/3}$ with $c_0=4\pi\xi(3/2)^{-2/3}$ is the critical temperature in the
non-interacting (ideal) Bose gas. 

This model confirms this prediction: in the accompanying paper \cite{NapReuSol2-15} we prove
that 
\begin{equation}
\label{tc20}
T_{\rm{c}}=T_{\rm{fc}}(1+h_1(\nu)\rho^{1/3}a+o(\rho^{1/3}a)), 
\end{equation}
where $\nu=\widehat{V}(0)/a$ and $\lim_{\nu\to8\pi}h_1(\nu)=1.49$. This result is in close
agreement with numerics: Monte Carlo methods suggest \cite{Arnold,Kash,NhoLan-04} that $c\approx 1.32$. In general $\nu>8\pi$, and it is believed, but not rigorously established, that the
Bogoliubov model is a good approximation if $\nu$ is replaced by
$8\pi$. The analysis leading to \eqref{tc20} can also be carried out in 2 dimensions; we discuss this in \cite{NapReuSol-17}.

Another issue is the asymptotic formula for the free energy (see
\cite{Sei-08,Yin-10} for the only rigorous results starting from the
full many-body problem). In \cite{NapReuSol2-15}, we provide formulas
for the free energy of a dilute Bose gas in different regions which
correspond to very low ($\rho a/T \gg 1$), fairly low ($\rho a/T \sim
O(1)$) and moderate ($\rho a/T \ll 1$) temperatures. In particular, if we let 
$\nu\to8\pi $, for very low temperatures we
reproduce the well-known Lee--Huang--Yang formula \bq \lim_{T\to 0}
F^{\rm{can}}(T,\rho)=4\pi a \rho^2 +\frac{512}{15}\sqrt{\pi}(\rho
a)^{5/2}+o(\rho a)^{5/2}. \nn \eq
For the reader's convenience, the main results of \cite{NapReuSol2-15} are also stated in the next section.\\
\medskip

\section{Main results and sketch of proof}
Throughout this
article, we assume that the two-body interaction potential and
its Fourier transform are radial functions that satisfy\footnote{As mentioned in the introduction, the case $V\equiv0$ is well-known and can be solved explicitly, so we exclude it by assumption. For completeness: if $V\equiv0$, a grand canonical minimizer exists only for $\mu\leq0$, $T\geq0$; a canonical one exists for all $\rho,T\geq0$; Theorem \ref{thm:BECvsSF} does not hold; Figure \ref{fig:phasediagram} is no longer accurate; the grand canonical phase transition of Theorem \ref{thm:phasetrangrandcan} is completely contained in the line $\mu=0$, where minimizers exist for all $\rho_0\in[0,\infty)$ and all $T\geq0$; the canonical phase transition of Theorem \ref{thm:phasetrancan} happens at the explicitly known \textit{free critical temperature} $T_{\rm{fc}}=c_0\rho^{2/3}$ with $c_0=4\pi\xi(3/2)^{-2/3}$.}
\begin{eqnarray}\label{positiveinteraction}
  V\geq 0,\quad \widehat{V}\geq0,\quad V\not\equiv0.
\end{eqnarray}
Moreover, we assume that 
\bq\label{interactionassumptions}
\widehat{V}\in C^{1}(\mathbb{R}^3),\ \widehat{V}\in L^1(\mathbb{R}^3),\ \|\widehat{V}\|_{\infty} <\infty,\ \|\nabla\widehat{V}\|_2<\infty,\ \|\nabla\widehat{V}\|_{\infty}<\infty.
\eq
\subsection{Existence of minimizers}\label{ssec:ex}
We start by providing the existence results that form the basis of any further analysis. 

\begin{theorem}[Existence of grand canonical minimizers for $T>0$]\label{thm:existencepositiveT}
  Let $T>0$. Assume the interaction potential is a radial function that satisfies \eqref{positiveinteraction} and
  \eqref{interactionassumptions}. Then there exists a minimizer for
  the Bogoliubov free energy functional
  \eqref{def:grandcanfreeenergyfunctional} defined on $\mathcal{D}$.
\end{theorem} 

It turns out that we need to assume some additional
regularity on the interaction potential to prove a similar statement for $T=0$ .

\begin{theorem}[Existence of grand canonical minimizers for $T=0$]\label{thm:existencezeroT}
  Assume the interaction potential fulfils the assumptions of Theorem
  \ref{thm:existencepositiveT}. If we assume in addition that
  $\widehat{V}\in C^{3}(\mathbb{R}^3)$ and that all derivatives of
  $\widehat{V}$ up to third order are bounded, then there exists a
  minimizer for the Bogoliubov free energy functional
  \eqref{def:grandcanfreeenergyfunctional} defined on $\mathcal{D}$
  for $T=0$.
\end{theorem}

We expect that our assumptions on the interaction potential are far
from optimal. A natural direction for further research would be to try
to extend the above results to the case of more singular
potentials. In the fermionic case, the existence of minimizers for the
HFB functional with Newtonian interaction turned out to be
surprisingly difficult to prove \cite{LenLew-10}.

\begin{remark}
  We would like to stress that the minimizers need not be unique. In
  fact, a detailed analysis of the dilute limit case in
  \cite{NapReuSol2-15} shows that there exist combinations of $\mu$
  and $T$ for which the problem \eqref{def:grandcanonicalminimization}
  has two minimizers with two different densities.
\end{remark}

We have analogous results in the canonical setting.

\begin{theorem}[Existence of canonical minimizers for
  $T>0$]\label{thm:existencecanonicalepositiveT}
  Let $T>0$. Assume the interaction potential is a radial function that satisfies \eqref{positiveinteraction} and
  \eqref{interactionassumptions}. Then the variational problem
  \eqref{def:canonicalminimization} admits a minimizer.
\end{theorem}

\begin{theorem}[Existence of canonical minimizers for $T=0$]\label{thm:existencecanonicalzeroT}
  Assume the interaction potential fulfils the assumptions of Theorem
  \ref{thm:existencecanonicalepositiveT}. If we assume in addition
  that $\widehat{V}\in C^{3}(\mathbb{R}^3)$ and that all derivatives of
  $\widehat{V}$ up to third order are bounded, then there exists a
  minimizer for the canonical minimization problem
  \eqref{def:canonicalminimization} at $T=0$.
\end{theorem}

A nice property that follows from the Euler--Lagrange equations and a trial state argument used in Theorems
\ref{thm:existencezeroT} and \ref{thm:existencecanonicalzeroT} is the fact that minimizers satisfy $\alpha^2=\gamma(\gamma+1)$ at $T=0$. This implies the following corollary (see the end of Subsection \ref{subsect:gcT=0}, just before Subsection \ref{thecancaselabel}).

\begin{corollary}[Structure of $T=0$ minimizers]
\label{cor:structureminimizerT=0}
Minimizers for the canonical and grand canonical problem at $T=0$ are pure quasi-free states.
\end{corollary}

Note that this result is not obvious. It is well known that pure quasi-free states are minimizers for quadratic Hamiltonians; it is also known that minimization problems over all quasi-free states can be restricted to pure quasi-free states at $T=0$ \cite{quasifree}, but our minimization does not correspond to a quadratic Hamiltonian and is not over all quasi-free states.

\subsection{Existence and structure of phase transition}
\label{ssec:exispt}
We now analyse the structure of the minimizers. Our first result
shows that Bose--Einstein condensation and the presence of pairing are equivalent within our models.

\begin{theorem}[Equivalence of BEC and pairing]\label{thm:BECvsSF}
  Let
  $(\gamma,\alpha,\rho_0)$ be a minimizing triple for
 either \eqref{def:grandcanfreeenergyfunctional} or \eqref{def:canonicalfreeenergyfunctional}. Then \bq \rho_0 =0
  \Longleftrightarrow \alpha \equiv 0. \nn \eq
\end{theorem}

Hence, there exists only one kind of phase transition within our
model. The next results show that this phase transition indeed exists.

\begin{theorem}[Existence of grand canonical phase transition]\label{thm:phasetrangrandcan}
  Let $\mu>0$.  Then there exist temperatures $0<T_1<T_2$ such that
  a minimizing triple $(\gamma,\alpha,\rho_0)$ of
  \eqref{def:grandcanonicalminimization} satisfies
\begin{enumerate}
\item $\rho_0=0$ for $T\geq T_2$
\item $\rho_0>0$ for $0\leq T \leq T_1$.
\end{enumerate}
\end{theorem} 

\begin{theorem}[Existence of canonical phase transition]\label{thm:phasetrancan}
For fixed  $\rho>0$ there exist temperatures $0<T_3<T_4$ such that a minimizing triple
$(\gamma,\alpha,\rho_0)$ of \eqref{def:canonicalminimization} satisfies 
\begin{enumerate}
\item $\rho_0=0$ for $T\geq T_4$
\item $\rho_0>0$ for $0\leq T \leq T_3$.
\end{enumerate}
\end{theorem} 

\begin{remark}
All the statements remain true in one and two dimensions.
\end{remark}

\subsection{Grand canonical phase diagram}\label{ssec:gcpd} The
results stated above together with their proofs allow us to sketch a
phase diagram of the system, see Figure
\ref{fig:phasediagram}.

\begin{figure}[h] 
\includegraphics[scale=0.3]{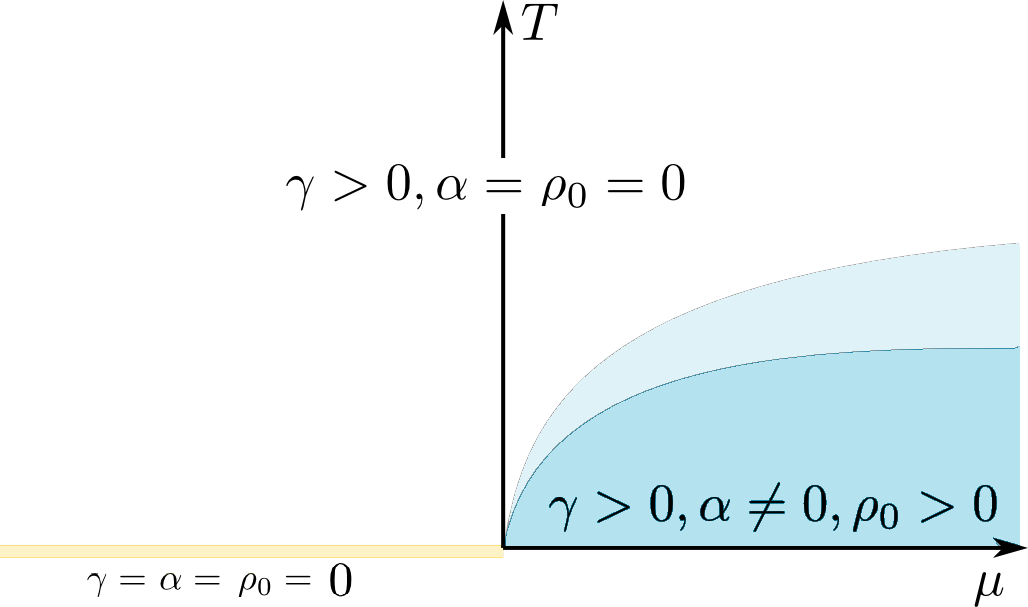}
\caption{The grand canonical phase diagram of the model. No diluteness is assumed. At $\mu\leq0$ and $T=0$, all quantities are zero, in particular there is no BEC. Increasing $T$ does not lead to a phase transition, although $\gamma$ becomes non-zero. For $\mu>0$ fixed and $T=0$, there is BEC. This remains the case when $T$ increases (darkest region), eventually leading to a phase transition somewhere in the lighter region before we enter the white region where $\rho_0=0$.}\label{fig:phasediagram}
\end{figure}

Note that at $T=0$ and $\mu<0$ the minimizer corresponds to the
vacuum.  Also, for negative chemical potentials there is no phase
transition in the system.

The area with the lighter shade of blue indicates that we cannot
rule out multiple phase transitions with different critical
temperatures. This is, however, unexpected. The vanishing of this
area as $\mu$ approaches zero from the right is a consequence of the
results in \cite{NapReuSol2-15}, which we review next. See, in particular, 
Theorem~\ref{thm:grandcancrittemp}.

\subsection{Main results of \cite{NapReuSol2-15}}
The main results of \cite{NapReuSol2-15} hold under several general
assumptions. For the following three results, we assume that the gas is dilute
\begin{eqnarray}\label{dillim}
\rho^{1/3}a\ll 1,
\end{eqnarray}
where $a$ is the scattering length of the potential. Furthermore, we define the constant $C$ by
\begin{equation}
\label{assumptionsV}
\int\widehat{V}\leq Ca^{-2}\hspace{1cm}\text{and}\hspace{1cm}\|\partial^n\widehat{V}\|_{\infty}\leq C a^{n+1}\text{\ for\ } 0\leq n\leq3,
\end{equation}
where $\partial^n$ is shorthand for all $n$-th order partial derivatives. With this definition, our estimates depend only on $C$ and not on $a$.
Recall $\nu:=\widehat{V}(0)/a$. The following theorems contain information about the critical temperature of the phase transition in the dilute limit. 

\begin{theorem}[Canonical critical temperature]
\label{thm:cancrittemp}
Let $(\gamma,\alpha,\rho_0)$ be a minimizing triple of \eqref{def:canonicalminimization} at temperature $T$ and density $\rho$. There is a monotone increasing function $h_1:(8\pi,\infty)\to\mathbb{R}$ with $h_1(\nu)\geq\lim_{\nu\to8\pi}h_1(\nu)=1.49$ such that
\begin{enumerate}
\item $\rho_0\neq 0$ if $T<T_{\rm fc}\left(1+h_1(\nu)\rho^{1/3}a+o(\rho^{1/3}a)\right)$
\item $\rho_0= 0$ if $T>T_{\rm fc}\left(1+h_1(\nu)\rho^{1/3}a+o(\rho^{1/3}a)\right)$,
\end{enumerate}
where $T_{\rm{fc}}=c_0\rho^{2/3}$ with $c_0=4\pi\xi(3/2)^{-2/3}$ is the critical temperature of the free Bose gas.
\end{theorem}

\begin{theorem}[Grand canonical critical temperature]
\label{thm:grandcancrittemp}
Let $(\gamma,\alpha,\rho_0)$ be a minimizing triple of \eqref{def:grandcanonicalminimization} at temperature $T$ and chemical potential $\mu$. There is a function $h_2:(8\pi,\infty)\to\mathbb{R}$ with $\lim_{\nu\to8\pi}h_2(\nu)=0.44$ such that
\begin{enumerate}
\item $\rho_0\neq 0$ if $T<\left(\frac{\sqrt{\pi}}{2\zeta(3/2)}\frac{8\pi}{\nu}\right)^{2/3}\left(\frac{\mu}{a}\right)^{2/3}+h_2(\nu)\mu+o(\mu)$
\item $\rho_0= 0$ if $T>\left(\frac{\sqrt{\pi}}{2\zeta(3/2)}\frac{8\pi}{\nu}\right)^{2/3}\left(\frac{\mu}{a}\right)^{2/3}+h_2(\nu)\mu+o(\mu)$.
\end{enumerate}
\end{theorem}
We refer to \cite{NapReuSol2-15} for the proof of these statements.\\

The second main result of \cite{NapReuSol2-15} provides an expansion of the canonical free energy \eqref{def:canonicalminimization} in the dilute limit. Here, we only state what happens for $\nu=\widehat{V}(0)/a\to8\pi$, which is a corollary of Theorem 10 in \cite{NapReuSol2-15}. We need to define an integral first:
\[
\begin{aligned}
I(s)&=(2\pi)^{-3}\int\ln\left(1-e^{-\sqrt{p^4+16\pi p^2s^2}}\right)dp.\\
\end{aligned}
\]

\begin{theorem}[Canonical free energy expansion for $\nu\to8\pi$]
\label{thm:canfreeenexp}
Let $T_{\rm{fc}}=c_0\rho^{2/3}$ be the critical temperature of the free Bose gas, and $\rho_{\rm{fc}}=(T/c_0)^{3/2}$ its corresponding critical density.
In the limit $\nu\to8\pi$, the canonical free energy  \eqref{def:canonicalminimization} can be expanded in the following way.
\begin{enumerate}
\item If $T>T_{\rm fc}\left(1+1.49\rho^{1/3}a+o(\rho^{1/3}a)\right)$, then
$$F^{\rm can }(T,\rho)=F_0(T,\rho)+\widehat{V}(0)\rho^2+O(\rho a)^{5/2},$$
and we have $\rho_\gamma=\rho$, $\rho_0=0$ for the minimizer. Here, $F_0(T,\rho)$ is the free energy of the non-interacting gas at density $\rho$ and temperature $T$.
\item If $T<T_{\rm fc}\left(1+1.49\rho^{1/3}a+o(\rho^{1/3}a)\right)$, then
\begin{equation*}
\begin{aligned}
F^{\rm{can}}(T,\rho)&=4\pi a\rho^2+4\pi a\rho_{\rm fc}^2+\frac{512}{15}\sqrt{\pi}(\rho a)^{5/2}\\
&\quad+T^{5/2}I\left(\sqrt{\frac{(\rho-\rho_{\rm fc}) a}{T}}+\frac{1}{4\sqrt{\pi}}\sqrt{T}a\right)-4\sqrt{\pi}T\rho_{\rm fc} a\sqrt{(\rho-\rho_{\rm fc}) a}\\
&\quad+o\left(T(\rho a)^{3/2}+(\rho a)^{5/2}\right).
\end{aligned}
\end{equation*}
\end{enumerate}
The last expression reduces to the Lee--Huang--Yang formula for $T\ll\rho a$:
\[
F^{\rm can}(T,\rho)=4\pi a\rho^2+\frac{512}{15}\sqrt{\pi}(\rho a)^{5/2}+o(\rho a)^{5/2}.
\]
\end{theorem}

\subsection{Sketch of proofs and set-up of the paper} The rest
of the paper is devoted to the proofs of the statements described in
Subsections \ref{ssec:ex}, \ref{ssec:exispt} and \ref{ssec:gcpd}.

In Section \ref{sec:preliminaries} we provide some general facts that
will be useful throughout the paper.

Section \ref{sec:exis_min_proof_T>0} contains the proofs of
Theorems \ref{thm:existencepositiveT} and
\ref{thm:existencecanonicalepositiveT}. Section
\ref{sec:exis_min_proof_T=0} provides proofs of Theorems
\ref{thm:existencezeroT} and \ref{thm:existencecanonicalzeroT}.

The proof of the existence of minimizers in our model is 
harder than it is for the fermionic BCS functional
\cite{HaiHamSeiSol-08}, mainly because of the occurrence of Bose--Einstein Condensation (BEC). Loosely
speaking, at sufficiently low temperatures bosons tend to
macroscopically occupy the same quantum state, which suggests that there is no a priori bound on the momentum distribution $\gamma(p)$. Therefore a minimizing sequence could converge to a
measure that has a singular part representing the
condensate. This scenario, however, has been included in the
construction of the functional by introducing 
the parameter $\rho_0$, which already describes the condensate. The situation is indeed simpler for fermions as there is
an a priori bound on $\gamma$: the Pauli principle says $\gamma\leq 1$.

We now present the main ideas behind the proof. As highlighted by \eqref{def:canonicalminimization} and \eqref{funcf}, the canonical minimization problem can be seen as a minimization in $(\gamma,\alpha)$ followed by one in $\rho_0$. Analysing the function $f$ defined in \eqref{funcf} in Subsection \ref{reduc}, we conclude that there is a particular $\rho^{\rm min}_0$ for which the minimum is attained. We then consider the canonical functional with fixed $\rho_0=\rho^{\rm min}_0$ and $\rho_\gamma=\rho-\rho^{\rm min}_0$. The remaining minimization should be done over $\gamma$ with this particular $\rho_\gamma$. To deal with this, we introduce a Lagrange multiplier in Subsection \ref{Leg}. As it is still too hard to directly prove the existence of a minimizing pair $(\gamma,\alpha)$ for this problem, we restrict our functional to the smaller space
\begin{align*}
\mathcal{M}_{\kappa}=\{(\gamma,\alpha)\in \cD'|\gamma(p)\leq \frac{\kappa}{p^2}\}. 
\end{align*}
This imposes an artificial bound that can now be used to prove the existence of minimizers for this restricted problem with standard techniques. The idea is then to construct
a minimizing sequence of the unrestricted problem out of the
minimizers $\gamma_\kappa$ of the restricted problem in the limit
$\kappa\rightarrow \infty$.

To this end, we prove several bounds for the $\gamma_\kappa$'s in Subsections \ref{ap1} and \ref{apgr}. In particular, we show that mass concentration (or condensation) of the minimizing sequence can only occur at $p=0$. The last main step in Subsection \ref{eoum} is then to show that this is actually impossible, because it would increase the energy compared to a solution where the mass would have been added to $\rho_0$ from the start, contradicting the fact that $\rho^{\rm min}_0$ was the minimizing $\rho_0$. 

The proof sketched above only works for $T>0$ since the bounds mentioned are not uniform in $T$ and deteriorate as $T\rightarrow
0.$ It turns out that, under suitable assumptions, the positive temperature minimizers $(\gamma^T,
\alpha^T,\rho_0^T)$ form a uniformly equicontinuous family that is
also a minimizing sequence for the $T=0$ problem. Using the
Arzel\`a--Ascoli theorem one can then extract a minimizer.

Proofs of Theorems \ref{thm:phasetrangrandcan} and
\ref{thm:phasetrancan} as well as Theorem
\ref{thm:BECvsSF} and the discussion of the grand canonical phase
diagram are provided in Section \ref{sec:exi_phase_tran}. 
Appendix \ref{app:derivation} contains an introduction to Bogoliubov's variational theory and
a derivation of the functional.

\section{Preliminaries}\label{sec:preliminaries}
Let us start with several remarks and bounds which will be used later. Throughout the proofs $C, C_1,...$ stand for
unspecified universal constants. 

Recall the notation
\begin{equation*}
\beta(p):=\sqrt{(\frac{1}{2}+\gamma(p))^{2}-\alpha(p)^{2}}. \label{beta}
\end{equation*}
Since
$$ s(\beta)=-(\beta-\frac12)\ln(\beta-\frac12)+(\beta+\frac12)\ln(\beta+\frac12)$$
and
$$\frac{\partial s(\beta)}{\partial \beta}=\ln\frac{\beta+\frac12}{\beta-\frac12},$$
we have 
\begin{equation}
s(\gamma,\alpha)<s(\gamma,0), \text{ if }\alpha\ne0. \label{ineq:sga<sg0}
\end{equation}

Several bounds will rely on the decomposition
$\alpha=\alpha_{<}+\alpha_{>}$ where
$\alpha_{<}:=\alpha\mathbbm{1}_{\{\gamma<1\}}$ and
$\alpha_{>}:=\alpha\mathbbm{1}_{\{\gamma\geq1\}}$. The condition \bq
\alpha^2\leq\gamma^2+\gamma \label{eq:gammaalphaconditions} \eq then
implies that $|\alpha_{>}|\leq\sqrt{2}\gamma$ and
$|\alpha_{<}|\leq\sqrt{2}\sqrt{\gamma}$. Thus using the assumptions on
$V$ and $\widehat{V}$ we have
\begin{equation}
\label{finaleqref?} 
\begin{aligned}
&\left|\int\widehat{V}(p)\alpha(p)dp\right|\leq\int\limits_{\{\gamma\geq1\}}\widehat{V}(p)|\alpha_{>}(p)|dp+\int\limits_{\{\gamma<1\}}\widehat{V}(p)|\alpha_{<}(p)|dp \\ 
&\leq \int\limits_{\{\gamma\geq1\}}\widehat{V}(p)\sqrt{2}\gamma(p)dp +\int\limits_{\{\gamma<1\}}\widehat{V}(p)\sqrt{2}\sqrt{\gamma(p)}dp \\ 
& \leq \sqrt{2}\widehat{V}(0)\int\limits_{\{\gamma\geq1\}}\gamma(p)dp+\sqrt{2}\int\limits_{\{\gamma<1\}}\widehat{V}(p)dp<C(\|\gamma\|_1,\widehat{V}).
\end{aligned}
\end{equation} 
Similarly
\bq
\begin{aligned} 
&\iint\widehat{V}(p-q)\alpha(p)\alpha(q)dpdq =\iint\widehat{V}(p-q)\alpha_{>}(p)\alpha_{>}(q)dpdq \\
&+\iint\widehat{V}(p-q)\alpha_{<}(p)\alpha_{<}(q)dpdq +
2\iint\widehat{V}(p-q)\alpha_{>}(p)\alpha_{<}(q)dpdq\\&<C(\|\gamma\|_1,\|\widehat{V}\|_1,\|\widehat{V}\|_\infty),
\end{aligned} \nn
\eq
since 
\begin{eqnarray*}
\iint\widehat{V}(p-q)\alpha_{>}(p)\alpha_{>}(q)dpdq\leq \widehat{V}(0)\left(\int|\alpha_>|\right)^2\leq 2\widehat{V}(0)\left(\int\gamma\right)^2,
\end{eqnarray*}
\bq
\begin{aligned}
\iint\widehat{V}(p-q)\alpha_{<}&(p)\alpha_{<}(q)dpdq\leq 2\iint\widehat{V}(p-q)\sqrt{\gamma(p)}\sqrt{\gamma(q)}dpdq \\ &\leq \iint\widehat{V}(p-q)(\gamma(p)+\gamma(q))dpdq=2 \left(\int \widehat{V}\right)\left(\int\gamma\right)
\end{aligned} \nn
\eq
and
\[
\begin{aligned}
\iint\widehat{V}(p-q)\alpha_{>}(p)\alpha_{<}(q)dpdq\leq 2\iint\widehat{V}(p-q)\gamma(p)dpdq. 
\end{aligned}
\]
Also, 
\begin{equation}
\iint\widehat{V}(p-q)\alpha(p)\alpha(q)dpdq=\int V(x)|\check{\alpha}(x)|^2dx\geq 0. \label{alphaxspace}
\end{equation}
Another useful consequence of \eqref{eq:gammaalphaconditions} is the (pointwise) bound 
\begin{equation}
\gamma+\alpha \geq -\frac12. \label{gammaplusalpha}
\end{equation}
 For the convolution terms one easily sees that
\begin{align}
  \|\widehat{V}\ast\gamma\|_\infty & \leq \|\widehat{V}\|_\infty \|\gamma\|_1\leq 
  C(\|\gamma\|_1,\|\widehat{V}\|_1,\|\widehat{V}\|_\infty), \label{convVgamma} \\
  \|\widehat{V}\ast\alpha\|_\infty & \leq \|\widehat{V}\|_\infty
  \|\alpha_>\|_1+\|\widehat{V}\|_2 \|\alpha_<\|_2\leq C\left(\|\gamma\|_1,\|\widehat{V}\|_1,\|\widehat{V}\|_\infty
    \right).
\label{convValpha}
\end{align}
We will also use the following lower bound on the free energy of a non-interacting system
\begin{eqnarray}
  \int p^2\gamma-T\int s(\beta)&\geq& \int p^2 \gamma_0-T\int s(\gamma_0,0)
  \nonumber\\&=& T\int\ln\left(1-e^{\frac{-p^2}{T}}\right)dp>-C, \label{noninteractingbound}
\end{eqnarray}
where $\gamma_0=\left(\exp(p^2/T)-1\right)^{-1}$. This follows from
\eqref{ineq:sga<sg0} and a direct computation. In particular, all terms in $\mathcal{F}$ are bounded on $\mathcal{D}'$.

\section{Existence of minimizers for $T>0$}\label{sec:exis_min_proof_T>0}

\subsection{Reduction to a minimization in $(\gamma,\alpha)$}
\label{reduc}
We work with the canonical functional
\[
\begin{aligned}
\cF^{\text{can}}(\gamma,\alpha,\rho_{0})=&\int p^2\gamma(p)dp-T\int s(\gamma(p),\alpha(p))dp+\frac12\widehat{V}(0)\rho^{2} \\
&+\rho_{0}\int \widehat{V}(p)(\gamma(p)+\alpha(p))dp 
\\&+\frac{1}{2}\iint 
\widehat{V}(p-q)(\gamma(p)\gamma(q)+\alpha(p)\alpha(q))dpdq, \label{def:auxenergy}
\end{aligned}
\]
where $\rho=\rho_0+\int\gamma=\rho_0+\rho_\gamma$.
For notational simplicity, we have absorbed $(2\pi)^{-3}$ in every integral compared to our usual definition \eqref{def:canonicalfreeenergyfunctional}, so that the measure is really $(2\pi)^{-3}dp$. We will use the same convention for the real space measure $dx$, but not for one-dimensional measures $dt$ or $ds$. 

We mostly focus on the canonical result (Theorem \ref{thm:existencecanonicalepositiveT}), and only remark on the grand canonical result (Theorem \ref{thm:existencepositiveT}) at the end of this section. We recall from the introduction \eqref{funcf}, 
\begin{equation}
\label{def:auxminimization}
f(\lambda,\rho_0)= \inf_{\substack{({\gamma},{\alpha})\in\cD'\\
      \int\gamma=\lambda
    }}\cF^{\rm{can}}(\gamma,\alpha,\rho_0),
\end{equation}
and 
\begin{equation}
\label{reduction}
F^{\rm{can}}(T,\rho)=\inf_{0\leq\rho_0\leq\rho}f(\rho-\rho_0,\rho_0),
\end{equation}
so that the minimization in $(\gamma,\alpha)$ is followed by one in $\rho_0$
The following proposition shows that $f$ is continuous, which implies that \eqref{reduction} has a minimizer $\rho^{\rm min}_0$. This reduces the problem to proving the existence of a minimizing $(\gamma,\alpha)$ for $f(\rho-\rho^{\rm min}_0,\rho^{\rm min}_0)$, which will do in the next few subsections.

\begin{proposition}
\label{lem:auxfunctionalconvexity}
The functional $\cF^{\rm{can}}$ is jointly strictly convex in $(\gamma,\alpha)$ and $\cD'$ is a convex set.
In addition, $f(\lambda,\rho_0)$ is strictly convex in $\lambda$ and continuous as a function of two variables, so that the infimum \eqref{reduction} is attained at some $0\leq\rho^{\min}_0\leq\rho$.
\end{proposition}
\begin{proof}
\textit{Convexity.}
  First notice that the Hessian of the entropy function $-s(\gamma,\alpha)$ \eqref{entro} regarded as a
  function of $\gamma$ and $\alpha$ is positive definite. Since
  $V\geq0$, the expressions $\iint
  \widehat{V}(p-q)\gamma(p)\gamma(q)dpdq$ and $\iint
  \widehat{V}(p-q)\alpha(p)\alpha(q)dpdq$ are convex in $\gamma$ and
  $\alpha$ respectively. It follows that $\cF^{\rm{can}}$ is jointly strictly convex in
  $(\gamma,\alpha)$. Convexity of $\cD'$ follows from a simple
  calculation.
 These conclusions imply convexity of $f$ in $\lambda$, but it is even strictly convex because of the presence of the $\lambda^2=\rho_\gamma^2$-term in $\widehat{V}(0)\rho^2/2$.

\textit{Continuity.} 
Define 
\begin{align*}
  \cF^{\rm{conv}}({\gamma},{\alpha},\rho_0)=& \int p^2
  {\gamma}(p)dp-TS({\gamma},{\alpha})\\&
  +\frac12\iint\widehat{V}(p-q)({\gamma}(p)+\rho_0 \delta_0)({\gamma}(q)+\rho_0 \delta_0)dpdq\\&
    +\frac12\iint\widehat{V}(p-q)({\alpha}(p)+\rho_0 \delta_0)({\alpha}(q)+\rho_0 \delta_0)dpdq,
\end{align*}
where $\delta_0$ is a delta function. This functional is jointly convex in
$(\gamma,\alpha,\rho_0)$, so that 
\[
\begin{aligned}
  f^{\text{conv}}(\lambda,\rho_0)=\inf_{\substack{({\gamma},{\alpha})\in\cD'\\
      \int\gamma=\lambda
    }}\cF^{\rm{conv}}({\gamma},{\alpha},\rho_0) 
\end{aligned}
\]
is jointly convex in $\lambda$ and $\rho_0$. This implies continuity on $(0,\infty)\times(0,\infty)$, but not necessarily at the boundaries. 

We now consider the points of the form $(\lambda^*,0)$ with $\lambda^*>0$. By convexity, we have
\[
\lim_{(\lambda,\rho_0)\to(\lambda^*,0)}f^{\text{conv}}(\lambda,\rho_0)\leq f^{\text{conv}}(\lambda^*,0),
\]
where the limit on the left is independent of the way we approach the boundary point.
To show the opposite inequality, note that we can use approximate minimizers for $(\lambda^*,\rho_0)$ as trial states for $(\lambda^*,0)$ by plugging them in with $\rho_0=0$. In the limit $\rho_0\to0$, these trial states approximate the limit above, proving continuity at this boundary. 

The boundary $(0,\rho^*_0)$ with $\rho^*_0\geq0$ can be treated in the same way, where we use the estimates from Section \ref{sec:preliminaries} to treat the terms involving $\gamma$ and $\alpha$ as $\lambda\to0$.

It follows that
$f(\lambda,\rho_0)=f^{\rm{conv}}(\lambda,\rho_0)+\frac12\widehat{V}(0)(\lambda+\rho_0)^2-\rho_0^2\widehat{V}(0)$ is continuous as well.
\qed
\end{proof}

\begin{corollary}[Radiality of minimizers]\label{radialminimizer}
  A minimizer $(\gamma,\alpha)$ of \eqref{def:auxminimization}, if it exists, is radial:
\begin{align*}
(\gamma(Rp),\alpha(Rp))=(\gamma(p),\alpha(p))\,\,\, \text{for any}\,\,\, R\in SO(3).  
\end{align*}
\end{corollary}
\begin{proof} The strict convexity of $\cF^{\rm{can}}$ in $(\gamma,\alpha)$ implies that a minimizer is unique assuming it exists. The result follows since $p\mapsto(\gamma(Rp),\alpha(Rp))$ is a minimizer if $p\mapsto(\gamma(p),\alpha(p))$ is. 
\qed
\end{proof}

\subsection{Legendre transform in $\lambda$ and a restricted problem}
\label{Leg}
It now suffices to show the existence of a minimizing $(\gamma,\alpha)$ for $f(\rho-\rho^{\rm min}_0,\rho^{\rm min}_0)$. This will require considerable effort. 

We first study $f$ at a point $(\lambda,\rho_0)$. Only $\gamma$ with $\rho_\gamma=\lambda$ are included in this minimization problem, but by strict convexity (Proposition \ref{lem:auxfunctionalconvexity}) this constrained minimization problem is equivalent to the unconstrained problem
\begin{align}
\label{def:dualauxproblem}
  \widehat{f}(\delta,\rho_0)
  :=\inf_{(\gamma,\alpha)\in\cD'}\left(\left. \cF^{\rm{can}}(\gamma,\alpha,\rho_0)-\delta\int\gamma\right. \right)
  =:\inf_{(\gamma,\alpha)\in\cD'}\cF^{\rm{can}}_{\delta}(\gamma,\alpha,\rho_0),
\end{align}
where $\delta$ is chosen as the slope of a supporting hyperplane at
$\lambda$ for the strictly convex function $\lambda\mapsto f(\lambda,\rho_0)$.
Note that each $\lambda$ corresponds to a unique $\delta$ by strict convexity.

We now prove a simple property of this new problem.

\begin{lemma}
\label{lem:coercivity}
Let $\delta\in \mathbb{R}$. There exists a constant $C_T$ (bounded as $T\to0$) such that
\[
\mathcal{F}_\delta^{\rm{can}}(\gamma,\alpha,\rho_0)\geq \frac{1}{2}\int p^2\gamma(p)dp-\delta\rho_\gamma+\frac12\widehat{V}(0)\rho_\gamma^2-C_T.
\]
In particular, any minimizing sequence of \eqref{def:dualauxproblem} is bounded in
$L^1(\mathbb{R}^3,dp)$ and $L^1(\mathbb{R}^3,p^2 dp)$.
\end{lemma}
\begin{proof}
We note
\[
\frac12\iint\widehat{V}(p-q)({\alpha}(p)+\rho_0 \delta_0)({\alpha}(q)+\rho_0 \delta_0)dpdq\geq0.
\]
We also use \eqref{ineq:sga<sg0} and bound 
\[
\frac12\int p^2\gamma(p) dp-TS(\gamma,0)\geq -C_T
\]
in the same way as in \eqref{noninteractingbound}.
\qed
\end{proof}

Unfortunately, since $L^1$ is not reflexive, the boundedness of a minimizing sequence of \eqref{def:dualauxproblem} is not enough to extract a weakly-converging subsequence. The best we can do at this point is to introduce a cut-off in our problem:
\begin{align}
  \widehat{f}_{\kappa}(\delta,\rho_0)=\inf_{(\gamma,\alpha)\in\cM_\kappa}\cF^{\rm{can}}_{\delta}(\gamma,\alpha,\rho_0), \label{def:restricteddual}
\end{align}
where $\kappa\geq 1$ and
\begin{align}
  \mathcal{M}_{\kappa}=\{(\gamma,\alpha)\in \cD'|\gamma(p)\leq
  \frac{\kappa}{p^2}\}.  \label{def:setmkappa}
\end{align}
 It is now possible to prove the existence of minimizers with standard techniques.
\begin{proposition}
\label{prop:existencerestricteddualminimizer}
There exists a minimizer for the restricted problem \eqref{def:restricteddual}.
\end{proposition}
\begin{proof}
  \textit{Step 1.} Let $(\gamma_n,\alpha_n)$ be a minimizing sequence
  in $\cM_\kappa$. It follows from Lemma \ref{lem:coercivity} that
  $\|\gamma_n\|_1 <C$ and that $\int p^2\gamma_n(p)dp <C$ for some
  constant $C$ depending on $T$, $V$ and $\delta$, but independent of $n$.

  \textit{Step 2.} We claim that $(\gamma_n,\alpha_n)$ is a bounded
  sequence in $L^{s}(\R^3)\times L^{s}(\R^3)$ for
  $s\in(\frac65,\frac32)$. To this end consider the set
\begin{align*}
\mathcal{A}_n=\{p\,|\, \gamma_n(p)>1\}.
\end{align*}
Since 
\begin{align*}
\int_{\R^3}\gamma_n\geq \int_{\mathcal{A}_n}\gamma_n\geq |\mathcal{A}_n|,
\end{align*}
the condition $\|\gamma_n\|_1 < C$ implies that for any $n$ we have $|\mathcal{A}_n|<C$. Furthermore for $s\in (1,\frac32)$ we have 
\bq
\begin{aligned}
\|\gamma_n\|_s^s=\int_{\R^3\setminus\mathcal{A}_n}\gamma_n^s+\int_{\mathcal{A}_n}\gamma_n^s\leq \int_{\R^3\setminus\mathcal{A}_n}\gamma_n+ \int_{\mathcal{A}_n}\frac{\kappa^s}{p^{2s}}dp<C, \label{eq:uniformgamman}
\end{aligned}
\eq
where we used the restriction imposed by \eqref{def:setmkappa}. Indeed, by the previous step we have $\int_{\R^3\setminus\mathcal{A}_n}\gamma_n<C$.  To bound the last term we use the fact that $|\mathcal{A}_n|<C$. Then
\bq
\begin{aligned}
\int_{\mathcal{A}_n}\frac{\kappa^s}{p^{2s}}dp=\int_{\mathcal{A}_n \cap B(0,1)}\frac{\kappa^s}{p^{2s}}dp+\int_{\mathcal{A}_n \setminus B(0,1)}\frac{\kappa^s}{p^{2s}}dp\leq \int_{ B(0,1)}\frac{\kappa^s}{p^{2s}}dp+|\mathcal{A}_n| \kappa^s,
\end{aligned} \nn
\eq
which is bounded uniformly in $n$ for $s<\frac32$. Here, $B(0,1)$ denotes the unit ball centred at the origin. Let us now consider the bound on $\|\alpha_n\|_s$. Using \eqref{eq:gammaalphaconditions} we have
\bq
\begin{aligned}
\|\alpha_n\|_s^s= \int_{\R^3\setminus\mathcal{A}_n}|\alpha_n|^s+\int_{\mathcal{A}_n}|\alpha_n|^s\leq 2^{\frac{s}{2}}\int_{\R^3\setminus\mathcal{A}_n}\sqrt{\gamma_n}^s+ 2^{\frac{s}{2}}\int_{\mathcal{A}_n}\gamma_n^s. \nn
\end{aligned}
\eq
By \eqref{eq:uniformgamman}, the last term is bounded uniformly in $n$. To bound the other term, notice that by the uniform bound $\int p^2\gamma_n(p)dp<C$ it follows from H\"older's inequality that
\bq
\begin{aligned}
\int_{\R^3\setminus\mathcal{A}_n}\gamma_n^{\frac{s}{2}}&= \quad\int_{\left(\R^3\setminus\mathcal{A}_n\right)\cap B(0,1) }\gamma_n^{\frac{s}{2}}\quad+\quad
\int_{\left(\R^3\setminus\mathcal{A}_n\right)\setminus B(0,1) }\gamma_n^{\frac{s}{2}} \\ 
&\leq  C \quad+\quad\int_{\left(\R^3\setminus\mathcal{A}_n\right)\setminus B(0,1) }\frac{(p^2\gamma_n)^{\frac{s}{2}}}{p^s}dp \\ 
&\leq C + \left(\int_{\left(\R^3\setminus\mathcal{A}_n\right)\setminus B(0,1) }p^2\gamma_n dp\right)^{\frac{s}{2}}\left(\int_{\left(\R^3\setminus\mathcal{A}_n\right)\setminus B(0,1) }p^{\frac{2s}{s-2}}dp\right)^{\frac{2-s}{2}}.
\end{aligned} \nn
\eq
For $\frac65<q<\frac32$ a uniform bound follows.

\textit{Step 3.}  By the previous step, we can find a subsequence that
converges weakly, i.e.\ there exist $(\tilde{\gamma},\tilde{\alpha})\in
L^{s}(\R^3)\times L^{s}(\R^3)$ such that $(\gamma_{n},\alpha_{n})
\rightharpoonup (\tilde{\gamma},\tilde{\alpha})$ for
$s\in(\frac65,\frac32)$.  Using Mazur's Lemma we can replace the
sequence with convex combinations and get strong convergence and by
going to a further subsequence we can assume that the limit is
pointwise almost everywhere. As the functional is convex we still have
a minimizing sequence.  It follows from Fatou's Lemma that
$(\tilde{\gamma},\tilde{\alpha})\in \cM_\kappa$. To show that
$(\tilde{\gamma},\tilde{\alpha})$ is a minimizer we will prove that  
$$\liminf_{n\to\infty}\cF^{\rm{can}}_{\delta}(\gamma_{n},\alpha_{n},\rho_0)\geq
\cF^{\rm{can}}_{\delta}(\tilde\gamma,\tilde\alpha,\rho_0).
$$
Indeed, 
$$
\frac12 p^2\gamma_n(p)-Ts(\gamma_n(p),\alpha_n(p))\geq \frac12
p^2\gamma_n(p)-Ts(\gamma_n(p),0)\geq T\ln(1-e^{-p^2/2T})
$$
and the function on the right is integrable. Using the bound
\eqref{def:setmkappa} we also see that $(\frac12
p^2-\delta)\gamma_n(p)$ is bounded below by an integrable
function. The same is true for $\widehat{V}(p)(\gamma_n(p)+\alpha_n(p))$ using \eqref{gammaplusalpha}.
The remaining quadratic terms have positive integrands.  Hence the
result follows by Fatou's Lemma.
\qed
\end{proof}
\begin{remark}\label{rem:uniformboundenessgammakappa}
Lemma \ref{lem:coercivity} implies that there exists a uniform bound on $\|\gamma_{\kappa}\|_1$. This follows from the simple observation that $\widehat{f}(\delta,\rho_{0})\leq 0$. 
\end{remark}

\begin{remark}\label{rem:uniformboundenessgammakappa_dim=1-2}
  In one and two dimensions, the restriction defined in
  \eqref{def:setmkappa} has to be appropriately modified to obtain
  analogous results on the existence of minimizers for the restricted
  functional. In fact, one needs to assume $\gamma(p)\leq \kappa/p^m$
  with $m\in(\frac12,1)$ and $m\in(1,2)$ in one and two dimensions
  respectively.
\end{remark}

We can now try to approach the minimizer of \eqref{def:dualauxproblem} by letting $\kappa\rightarrow\infty$, but we first do some preparatory work in the next subsections.

\subsection{A priori bounds on $\gamma$ and $\alpha$}
\label{ap1}
First, we show that any potential minimizer $\gamma$ is strictly positive almost everywhere. 

\begin{lemma}[Positivity of $\gamma$]\label{lem:gammabigger0}
Let $T>0$, $\delta\in\R$, $\rho_0\geq0$, and $\kappa\geq T$. There exists a constant $C:= C\left(\delta,\rho_0,\|\gamma\|_1,\widehat{V}\right)$ such that for any minimizer $(\gamma,\alpha)$ of either \eqref{def:dualauxproblem} or \eqref{def:restricteddual}, the set
\begin{align*}
\mathcal{S}:=\left\{p\, | \, \gamma<e^{-\frac{p^2+C}{T}}\right\}.
\end{align*}
has zero measure.
\end{lemma}
\begin{proof} Since $\kappa\geq T$, we have $\kappa/p^2\geq e^{-p^2/T}$ and the upper bound defining 
$\mathcal{S}$ is within the restriction in \eqref{def:setmkappa}.
  The functional derivative \footnote{We use the notation where
    $\frac{\partial\mathcal{F}^{\rm{can}}_\delta}{\partial \gamma}$ is
    defined by
    $\int\frac{\partial\mathcal{F}^{\rm{can}}_\delta}{\partial
      \gamma}(p)\phi(p)dp=\left.\frac{d}{dt}\mathcal{F}^{\rm{can}}_\delta(\gamma+t\phi,\alpha)\right|_{t=0}$.}
  of $\cF^{\rm{can}}_{\delta}$ in $\gamma$ gives
\begin{equation}
\frac{\partial\mathcal{F}^{\rm{can}}_\delta}{\partial\gamma}=p^2-\delta+\rho\widehat{V}(0)+\rho_0 \widehat{V}(p) + \widehat{V}\ast \gamma(p)-T\frac{\gamma+\frac{1}{2}}{\beta}\ln{\frac{\beta+\frac{1}{2}}{\beta-\frac{1}{2}}}.\nn
\end{equation}
Since $\frac12\leq\beta=\sqrt{(\gamma+\frac12)^2-\alpha^2}\leq\gamma+\frac12$, we have
\begin{equation}
\frac{\partial\mathcal{F}^{\rm{can}}_\delta}{\partial\gamma}\leq
p^2+C-T\frac{\gamma+\frac{1}{2}}{\beta}\ln{\frac{\beta+\frac{1}{2}}{\beta-\frac{1}{2}}}\leq
p^2+C-T\ln{\frac{1}{\beta-\frac{1}{2}}},\nn
\end{equation}
where $C:= C\left(\delta,\rho_0,\|\gamma\|_1, \widehat{V}\right)$ follows from \eqref{convVgamma}. Thus 
\begin{align*}
\frac{\partial\mathcal{F}^{\rm{can}}_\delta}{\partial\gamma}<0
\end{align*}
if
\begin{equation*}
\beta<e^{-\frac{p^2+C}{T}}+\frac{1}{2}. 
\end{equation*}
Since $\beta\leq\gamma+\frac12$, this certainly holds whenever
\begin{equation*}
\gamma<e^{-\frac{p^2+C}{T}}.
\end{equation*}
Thus the functional derivative is negative for $p\in
\mathcal{S}$. Hence, if the set had positive measure, we would be able
to lower the free energy by increasing $\gamma$ on it. This would
contradict the assumption that $(\gamma,\alpha)$ is a minimizer.
\qed
\end{proof}
From now on we assume that $\kappa>T$. We now show that
$\alpha(p)^2<\gamma(p)\left(\gamma(p)+1\right)$ almost
everywhere for minimizers of \eqref{def:dualauxproblem} or
\eqref{def:restricteddual}. Note that the statement is vacuous if
$\gamma(p)=0$. This possibility is, however, excluded (almost
everywhere) by Lemma \ref{lem:gammabigger0}.

\begin{lemma}\label{lem:apriorialpha}
Let $T>0$, $\delta\in\mathbb{R}$ and $\rho_0\geq0$. There exists a constant $C:= C\left(\rho_0,\|\gamma\|_1, \widehat{V}\right)$ such that for any minimizer $(\gamma,\alpha)$ of either \eqref{def:dualauxproblem} or \eqref{def:restricteddual}, the set 
\begin{equation}
\mathcal{P}:=\left\{p\,|\, \alpha(p)^2>\gamma(p)\left(\gamma(p)+1\right)-e^{-C/(Tc(\gamma(p)))},\, \alpha(p)^2>\frac{1}{2}\gamma(p)\left(\gamma(p)+1\right)\right\} \nn
\end{equation}
has zero measure, where
\begin{equation*}
c(\gamma):=\sqrt{\frac{2\gamma(\gamma+1)}{2\gamma(\gamma+1)+1}}.
\end{equation*}
\end{lemma} 
\begin{proof}
The functional derivative of $\cF^{\rm{can}}_\delta$ in $\alpha$ gives
\begin{equation*}
\frac{\partial\mathcal{F}^{\rm{can}}_\delta}{\partial\alpha}=\rho_0 \widehat{V}(p) + \widehat{V}\ast \alpha(p)+T\frac{\alpha}{\beta}\ln{\frac{\beta+\frac{1}{2}}{\beta-\frac{1}{2}}}.
\end{equation*}
Assume first that $\alpha(p)^2>\frac{1}{2}\gamma(p)\left(\gamma(p)+1\right)$. Then
\begin{equation*}
\left|\frac{\alpha}{\beta}\right|\geq\sqrt{{\frac{\frac{1}{2}\gamma\left(\gamma+1\right)}{\gamma(\gamma+1)+\frac{1}{4}-\alpha^2}}}
\geq\sqrt{{\frac{\frac{1}{2}\gamma\left(\gamma+1\right)}{\frac{1}{2}\gamma(\gamma+1)+\frac{1}{4}}}}=c(\gamma).
\end{equation*}
Another estimate that holds by the assumptions on $V$ and inequality \eqref{convValpha} is 
\begin{equation*}
\left|\rho_0 \widehat{V}(p) + \widehat{V}\ast \alpha(p)\right|\leq \rho_0\|\widehat{V}\|_\infty+\|\widehat{V}\ast\alpha\|_\infty
\leq C\left(\rho_0,\|\gamma\|_1, \widehat{V}\right).
\end{equation*}
If $\alpha\geq\sqrt{\frac{1}{2}\gamma(p)(\gamma(p)+1)}$, then
\begin{equation*}
\frac{\partial\mathcal{F}^{\rm{can}}_\delta}{\partial\alpha}
\geq -C +T c(\gamma)\ln{\frac{\beta+\frac{1}{2}}{\beta-\frac{1}{2}}}
\geq -C +T c(\gamma)\ln{\frac{1}{\beta-\frac{1}{2}}},
\end{equation*}
where we have used $\beta\geq\frac{1}{2}$ in addition to the previous estimates.
Similarly, if $\alpha\leq-\sqrt{\frac{1}{2}\gamma(p)(\gamma(p)+1)}$, we estimate
\begin{equation*}
\frac{\partial\mathcal{F}^{\rm{can}}_\delta}{\partial\alpha}
\leq C -T c(\gamma)\ln{\frac{\beta+\frac{1}{2}}{\beta-\frac{1}{2}}}
\leq C - T c(\gamma)\ln{\frac{1}{\beta-\frac{1}{2}}}.
\end{equation*}

\noindent In the first/second case the derivative is positive/negative whenever
\begin{equation*}
\beta\leq\frac{1}{2}+e^{-C/(Tc(\gamma))},
\end{equation*}
and using the definition of $\beta^2$ we find that this happens when
\begin{equation*}
\alpha^2>\gamma(\gamma+1)-e^{-2 C/(Tc(\gamma))}-e^{-C/(Tc(\gamma))}.
\end{equation*}
This means that the derivative is positive for $p$ in $\mathcal{P}$ and $\alpha(p)\geq0$, and negative for $p$ in $\mathcal{P}$  and $\alpha(p)\leq0$. Hence, if the set had positive measure, we would be able to lower the energy by varying on it, which contradicts the assumption that $(\gamma,\alpha)$ is a minimizer.
\qed
\end{proof}

Since we already know that $\gamma(p)>0$ almost everywhere, this implies that $-\sqrt{\gamma(p)\left(\gamma(p)+1\right)}<\alpha(p)<\sqrt{\gamma(p)\left(\gamma(p)+1\right)}$ almost everywhere. Thus the Euler--Lagrange equation for $\alpha$ holds with equality for minimizers of both \eqref{def:dualauxproblem} and \eqref{def:restricteddual}.

\subsection{A priori bound for $\gamma$ in the restricted case}
\label{apgr}

The existence of minimizers for \eqref{def:restricteddual}, as well as
the a priori bounds established in the previous subsection give us
access to the Euler--Lagrange equations for the restricted
problem. Indeed, we have
\begin{align}
\frac{\partial\mathcal{F}^{\rm{can}}_\delta}{\partial\gamma}&=p^2-\delta+\rho\widehat{V}(0)+\rho_0 \widehat{V}(p) + \widehat{V}\ast \gamma(p)-T\frac{\gamma+\frac{1}{2}}{\beta}\ln{\frac{\beta+\frac{1}{2}}{\beta-\frac{1}{2}}} \nn \\
&=\left\{\begin{array}{ll}
\leq0 & \mbox{if }\gamma(p)=\kappa/p^2\\
=0 & \mbox{if }0\leq\gamma(p)< \kappa/p^2
\end{array}\right. \label{eq:ELrestrictegamma} \\
\frac{\partial\mathcal{F}^{\rm{can}}_\delta}{\partial\alpha}&=\rho_0 \widehat{V}(p) + \widehat{V}\ast \alpha(p)+T\frac{\alpha}{\beta}\ln{\frac{\beta+\frac{1}{2}}{\beta-\frac{1}{2}}}=0.
\label{eq:ELconstrainedalpha}
\end{align}

We now analyse these equations in order to derive a priori bounds for $\gamma_\kappa$, the minimizer of the restricted problem \eqref{def:restricteddual}. These bounds are then used to  show convergence of $\gamma_\kappa$ to a minimizer of the unrestricted problem \eqref{def:dualauxproblem} as $\kappa\rightarrow \infty$.

\begin{lemma}[Large $p$ a priori bound for $\gamma$]\label{lem:boundgammalargep}
  Let $T>0$, $\delta\in {\mathbb R}$, and $0\leq \rho_0\leq
  \rho$. There exist positive constants $P_0$ and $C$ such that for any minimizer $(\gamma,\alpha)$ of \eqref{def:restricteddual} with $\kappa>\max\{1,T\}$ and $|p|>P_0$, we have
\begin{equation}\label{eq:largepboundgamma}
\gamma(p)\leq C |p|^{-4}.
\end{equation}
Moreover, as $T$ goes to zero $P_0$ is uniformly bounded above and $C$ uniformly bounded above. 
\end{lemma}
\begin{proof}
Assume $\alpha\neq0$. Using \eqref{eq:ELconstrainedalpha}, Lemma~\ref{lem:coercivity} and \eqref{convValpha} we see that there is a $P_0$ such that for $|p|>P_0$ 
\begin{eqnarray}\label{estimatederiv}
\begin{aligned}
0&\geq\frac{\partial\mathcal{F}^{\rm{can}}_\delta}{\partial\gamma}
\geq \frac{1}{2}p^2-T\frac{\gamma+\frac{1}{2}}{\beta}\ln{\frac{\beta+\frac{1}{2}}{\beta-\frac{1}{2}}}
 \\ &=\frac{1}{2}p^2+\frac{\gamma+\frac{1}{2}}{\alpha}\left(\rho_0 \widehat{V}(p) + \widehat{V}\ast \alpha(p)\right)
\geq \frac{1}{2}p^2-C\frac{\gamma+\frac{1}{2}}{|\alpha|}.
\end{aligned}
\end{eqnarray}
Hence, we have
\begin{equation*}
\frac{\alpha^2}{\left(\gamma+\frac{1}{2}\right)^2}\leq Cp^{-4}.
\end{equation*}
Note that we can now drop the assumption $\alpha\neq0$ since the above also holds if $\alpha=0$. 
This implies
\begin{equation*}
\beta^2=\left(\gamma+\frac{1}{2}\right)^2\left(1-\frac{\alpha^2}{\left(\gamma+\frac{1}{2}\right)^2}\right)\ge\left(\gamma+\frac{1}{2}\right)^2\left(1-Cp^{-4}\right),
\end{equation*}
which can be rewritten as 
\begin{equation}
\frac{\gamma+\frac{1}{2}}{\beta}\leq \left(1-Cp^{-4}\right)^{-\frac12}.
\label{intermstep1}
\end{equation}
Returning to the second estimate in \eqref{estimatederiv}, we obtain
\begin{equation*}
0\geq\frac{\partial\mathcal{F}^{\rm{can}}_\delta}{\partial\gamma}\geq\frac{1}{2}p^2-T\left(1-Cp^{-4}\right)^{-\frac12}\ln{\frac{\beta+\frac{1}{2}}{\beta-\frac{1}{2}}}.
\end{equation*}
Rewriting this inequality leads to 
\begin{equation*}
\beta\leq\frac{1}{2}\frac{\exp\left[\frac{p^2}{2T}\left(1-Cp^{-4}\right)^{1/2}\right]+1}{\exp\left[\frac{p^2}{2T}\left(1-Cp^{-4}\right)^{1/2}\right]-1}
=\frac{1}{2}+\left[\exp\left[\frac{p^2}{2T}\left(1-Cp^{-4}\right)^{1/2}\right]-1\right]^{-1}.
\end{equation*}
Combining this with \eqref{intermstep1} we finally obtain
\begin{equation*}
\gamma+\frac{1}{2}\leq \left(1-Cp^{-4}\right)^{-1/2}\left(\frac{1}{2}+\left[\exp\left[\frac{p^2}{2T}\left(1-Cp^{-4}\right)^{1/2}\right]-1\right]^{-1}\right),
\end{equation*}
which is $\frac{1}{2}+ C p^{-4}+\mathcal{O}(p^{-8})$ for $p$ large enough.
\qed 
\end{proof}

Let $\gamma_\kappa$ be a minimizer for the restricted problem \eqref{def:restricteddual} for a given $\kappa$. We define
\begin{align}
p_\kappa:=\sup\left\{|p|\ \Big|\ \gamma_{\kappa}(p)=\frac{\kappa}{p^2}\right\}. \label{def:pkappa}
\end{align}
Note that the bound on $\gamma_\kappa$ for large $|p|$ proved in Lemma
\ref{lem:boundgammalargep} implies that $p_\kappa$ cannot be infinite
for any $\kappa$. We will therefore assume henceforth that $p_\kappa$ is
finite. It could be that $p_\kappa=-\infty$ (in case the set
\eqref{def:pkappa} is empty), but then our proof works as well.

We shall now work towards a priori bounds on $\gamma$ for small
$p$. We start by proving a lemma that we will use twice later on.
\begin{lemma}
  Let $a>0$ and $g:[a,\infty)\to[0,\infty)$ be a non-negative, continuously differentiable
  function. If $|g'(t)|<C_a$
  for $a\leq t\leq2a$, then
\[
\int_{|p|\geq a} g(|p|)^{-1}d^3p\geq (2\pi^2)^{-1} a^2 {C_a}^{-1}\ln\left(1+\frac{C_a}{g(a)}a\right).
\]
\label{usefullemma1}
\end{lemma}
\begin{proof}
By assumption we have $g(|p|)\leq g(a)+C_a(|p|-a)$ for $a\leq|p|\leq2a$. 
Thus (recalling our convention for the measures $dp$ and $dt$ explained above \eqref{def:auxminimization})
\begin{align*}
\int_{|p|\geq a} g(|p|)^{-1}d^3p&\geq(2\pi^2)^{-1}\int^{2a}_a\left[g(a)+C_a(t-a)\right]^{-1}t^2dt\\
&\geq (2\pi^2)^{-1} a^2\int^{a}_0\left[g(a)+C_at\right]^{-1}dt\\
&=(2\pi^2)^{-1} a^2 {C_a}^{-1}\ln\left(1+\frac{C_a}{g(a)}a\right).
\end{align*}
\qed
\end{proof}
To obtain the desired bound for small $p$ we will apply this lemma to the radial function $\tilde{g}$ given by
\begin{equation}
\tilde{g}(|p|):=\gamma_\kappa(p)^{-1},
\label{helpfunction}
\end{equation}
where $\gamma_\kappa$ is a minimizer for the restricted problem (note this is indeed a radial function by arguments similar to those in Corollary \ref{radialminimizer}). This means we need to get a bound on the derivative of $\gamma_\kappa^{-1}$. In the calculations below we assume that $(\gamma,\alpha)$ is a minimizer for \eqref{def:restricteddual} for a fixed $\delta\in\mathbb{R}$ and $\rho_0\geq0$ (we drop the subscript $\kappa$ for convenience).

We start our analysis from the Euler--Lagrange equations, which hold with equality for $|p|>p_\kappa$ :
\begin{equation}
\label{ELgamma1}
\begin{aligned}
\frac{\partial\mathcal{F}^{\rm{can}}_\delta}{\partial\gamma}&=p^2-\delta+\rho\widehat{V}(0)+\rho_0 \widehat{V}(p) + \widehat{V}\ast \gamma(p)-T\frac{\gamma+\frac{1}{2}}{\beta}\ln{\frac{\beta+\frac{1}{2}}{\beta-\frac{1}{2}}}=0\\
\frac{\partial\mathcal{F}^{\rm{can}}_\delta}{\partial\alpha}&=\rho_0 \widehat{V}(p) + \widehat{V}\ast \alpha(p)+T\frac{\alpha}{\beta}\ln{\frac{\beta+\frac{1}{2}}{\beta-\frac{1}{2}}}=0.
\end{aligned}
\end{equation}
We define
\begin{eqnarray} 
\begin{aligned}
\label{def:A(p)-B(p)}
A(p)&:=\frac{1}{T}\left(p^2-\delta+\rho\widehat{V}(0)+\rho_0 \widehat{V}(p) + \widehat{V}\ast \gamma(p)\right)\\
B(p)&:=\frac{1}{T}\left(\rho_0 \widehat{V}(p) + \widehat{V}\ast \alpha(p)\right). 
\end{aligned}
\end{eqnarray}
By squaring, subtracting and taking a square root we obtain
\begin{equation*}
\ln{\frac{\beta+\frac{1}{2}}{\beta-\frac{1}{2}}}=\sqrt{A(p)^2-B(p)^2}=:G(p),
\end{equation*}
and
\begin{equation*}
\beta=\frac{1}{2}\frac{e^G+1}{e^G-1}.
\end{equation*}
Combined with \eqref{ELgamma1} this leads to
\begin{equation}
\gamma=\beta\left[\ln{\frac{\beta+\frac{1}{2}}{\beta-\frac{1}{2}}}\right]^{-1}A-\frac{1}{2}
=\frac{1}{2}\left[\frac{e^G+1}{G(e^G-1)}A-1  \right].
\label{gammaexp}
\end{equation}
We know that $\gamma\leq \kappa/p^2$ and also that the expression above is correct for $|p|>p_\kappa$. Therefore the denominator cannot go to zero in this region (implying $G$ cannot go to zero). 
Combining this with the relation between $G$ and $A$ we obtain
\begin{equation}
A\geq G>0.  \label{AbiggerG}
\end{equation}
Together with \eqref{gammaexp}, this implies
\begin{equation}
\gamma^{-1}=\frac{2G(e^G-1)}{e^G(A-G)+A+G}\leq (e^G-1)\leq (e^A-1).
\label{gammaderiv}
\end{equation}
Recall the definition \eqref{helpfunction}. From \eqref{gammaderiv} it
follows that on $(p_k,\infty)$ we have
\begin{align*}
  \frac{\tilde{g}'(|p|)}{2}=&\frac{G'\left[e^G-1+Ge^G\right]}{D}\\&-\frac{G(e^G-1)\left[G'e^G(A-G)+A'(e^G+1)+G'(1-e^G)\right]}{D^2}\\
  =&\frac{1}{D}\left[GG'\left(\frac{e^G-1}{G}+e^G\right)\right]-GG'\frac{G}{D}\frac{e^G-1}{G}e^G\left[\frac{A}{D}-\frac{G}{D}\right]\\
  &-A'\left(\frac{G}{D}\right)^2\frac{e^G-1}{G}\left(e^G+1\right)+GG'\left(\frac{e^G-1}{G}\right)^2\left(\frac{G}{D}\right)^2,
\end{align*}
where $D=e^G(A-G)+A+G$.  Note that if $P>p_\kappa$ is bounded then all terms except the first are
bounded on $p_k<|p|\leq P$: $A$ and $A'$ are bounded by their form
\eqref{def:A(p)-B(p)} and the assumptions on $V$; $G$ is bounded by
$A$; $D$ is bigger than or equal to both $G$ and $A$. The boundedness of
$GG'$ follows from
\begin{equation*}
G'=\frac{1}{G}\left(AA'+BB'\right)
\end{equation*}
and the boundedness from all terms between the brackets ($B$ and $B'$ are bounded for similar reasons as $A$ and $A'$). It follows that
\begin{equation*}
\tilde{g}'(p)\leq \frac{C_2}{D(p)}+C_3\leq \frac{C_2}{A(p)}+C_3,
\end{equation*}
where the $C_i:=C_i(P,\|\widehat{V}\|_\infty,\|\nabla\widehat{V}\|_\infty, \|\gamma\|_1)$ are constants. To obtain a final bound on $\tilde{g}'(p)$ we need the following lemma.
\begin{lemma}
For $p_\kappa<|p|\leq\frac{1}{2}P$, where $P>2p_\kappa$ is a given constant, there exists $C_1:=C_1(P,\|\widehat{V}\|_\infty,\|\nabla\widehat{V}\|_\infty, \|\gamma\|_1)$ such that
\[
A(p)\geq \ln \left[1+C_1 |p| e^{-\frac{2\pi^2C_1 \|\gamma\|_1}{p^2}}\right].
\]
\label{firstapplemma}
\end{lemma}
\begin{proof}
In order to apply Lemma \ref{usefullemma1} we define the function
\[
g(|p|):=e^{A(p)}-1. 
\]
By \eqref{def:A(p)-B(p)}, \eqref{AbiggerG} and our assumptions on $V$
it follows that $g$ is positive and continuously
differentiable. To apply Lemma \ref{usefullemma1}, we need a bound on
its derivative for $|p|\in(p_\kappa,P)$. Since $A$ and $A'$ are
bounded for $p_\kappa <|p|\leq P$ we have
\begin{equation*}
|g'(|p|)|=|A'(|p|)|e^{A(|p|)}\leq C_1(P,\|\widehat{V}\|_\infty,\|\nabla\widehat{V}\|_\infty, \|\gamma\|_1).
\end{equation*}
Using Lemma \ref{usefullemma1} and \eqref{gammaderiv}, we now get for $p_\kappa <|p|\leq\frac{1}{2}P$
\begin{equation*}
\|\gamma\|_1\geq \int_{|\xi|\geq |p|} g(|\xi|)^{-1}d^3\xi\geq (2\pi^2)^{-1} p^2C_1^{-1}\ln \left[1+C_1 |p|( e^{A(p)}-1)^{-1}\right].
\end{equation*}
Rewriting this proves the lemma.
\qed
\end{proof}
It follows that
\begin{equation}
\tilde{g}'(p)\leq\frac{C_2}{\ln \left[1+C_1 |p| e^{-\frac{2\pi^2C_1 \|\gamma\|_1}{p^2}}\right]}+C_3=:\eta(|p|).
\label{boundftilde}
\end{equation}
Since the function $\eta$ is decreasing, we can bound it on the interval $[|p|,2|p|]$ by its value at $|p|$.

\begin{lemma}[Small $p$ a priori bound for $\gamma$]
For $p_k<|p|\leq P_0$ (where $P_0$ was defined in Lemma~\ref{lem:boundgammalargep}), we have 
\begin{equation}
\gamma(p)\leq |p|^{-1}\eta(|p|)^{-1}e^{\frac{2\pi^2\|\gamma\|_1\eta(|p|)}{p^2}},
\label{eq:smallpboundgamma}
\end{equation}
where $\eta$ is defined in \eqref{boundftilde} with $P$ replaced by $2P_0$ (in the dependence of the constants).
\label{lem:boundgammasmallp}
\end{lemma}

\begin{proof}
Equation \eqref{boundftilde} gives us the bound required to apply Lemma \ref{usefullemma1}.
We therefore get 
\[
\|\gamma\|_1\geq \int_{|\xi|\geq |p|} \tilde{g}(|\xi|)^{-1}d^3\xi
\geq (2\pi^2)^{-1}p^2\eta(|p|)^{-1}\ln \left[1+|p|\eta(|p|)\gamma(p)\right],
\]
which gives the stated result upon rewriting.
\qed
\end{proof}

\begin{remark}
  All a priori bounds derived in this subsection remain (up to minor
  modifications) true in one and two dimensions.
\end{remark}

Equipped with these bounds we shall move towards the proof of
existence of a minimizer of the dual problem \eqref{def:dualauxproblem} (at the relevant $\delta^{\rm min}$ and $\rho_0^{\rm min}$).

\subsection{Existence of unrestricted minimizers}
\label{eoum}
We are now ready to prove Theorem \ref{thm:existencecanonicalepositiveT} by showing the existence of minimizing $(\gamma,\alpha)$ for $f(\rho-\rho_0^{\text{min}},\rho_0^{\text{min}})$ \eqref{def:auxminimization}, where $\rho_0^{\text{min}}$ was introduced in Proposition \ref{lem:auxfunctionalconvexity}. By the strict convexity of $f$ in $\lambda$, this is equivalent to finding a minimizer of its Legendre transform $\widehat{f}(\delta^{\text{min}},\rho_0^{\text{min}})$ \eqref{def:dualauxproblem}, where $\delta^{\rm min}\in\mathbb{R}$ is in one-to-one correspondence with $\rho-\rho_0^{\text{min}}$. We will denote $\lambda^{\rm min}:=\rho-\rho^{\rm min}_0$. The goal of this subsection is thus to prove the following result.
\begin{proposition}
\label{thm:existenceunconstrainedminimizers}
There exists a minimizer $(\tilde{\gamma},\tilde{\alpha})$ for \eqref{def:dualauxproblem} with $\delta=\delta^{\rm{min}}$, $\rho_0=\rho_0^{\rm{min}}$ and $T>0$, and therefore of \eqref{def:auxminimization} with $\lambda=\rho-\rho^{\rm min}_0$ and $\rho_0=\rho^{\rm min}_0$.
\end{proposition}

As we have proved in Proposition \ref{prop:existencerestricteddualminimizer}, for given $\delta$, $\rho_0$ and $\kappa$ we can find $(\gamma_\kappa,\alpha_\kappa)$ that minimize the restricted problem \eqref{def:restricteddual}.
We would like to combine the bounds in Lemmas \ref{lem:boundgammalargep} and \ref{lem:boundgammasmallp} to extract a minimizer for \eqref{def:dualauxproblem} from the sequence $(\gamma_\kappa,\alpha_\kappa)$. To do this, we first need to prove that we can actually reach the whole of $|p|>0$ using the regions $|p|>p_\kappa$.
\begin{lemma}
There exists a subsequence of $(\gamma_\kappa,\alpha_\kappa)$ such that $p_\kappa\rightarrow0$ as $\kappa\rightarrow\infty$.
\end{lemma}
\begin{proof}
First note that the assumption that $\liminf_{\kappa\rightarrow\infty}p_\kappa=c>0$ together with Lemma \ref{lem:boundgammasmallp} and the uniform bound on $\|\gamma_\kappa\|_1$ will lead to a contradiction if we can show that
\begin{equation}
\lim_{|p|\searrow p_\kappa}\gamma_k(p)=\frac{\kappa}{p_\kappa^2}.
\label{gammaatpk}
\end{equation}
Indeed, in this situation the left-hand side of
\eqref{eq:smallpboundgamma} tends to infinity, whereas the right-hand
side is bounded yielding a contradiction.  We conclude that
$\liminf_{\kappa\rightarrow\infty}p_\kappa\leq0$ and hence we can
extract a subsequence that has $p_\kappa\rightarrow0$.

To prove \eqref{gammaatpk}, we first claim there exist $\tilde{\gamma},\tilde{\alpha}$ and associated $\tilde{\beta}$ such that  
\begin{eqnarray}
\begin{aligned}
\label{wanted11}
p^2-\delta+\rho\widehat{V}(0)+\rho_0 \widehat{V}(p) + \widehat{V}\ast \gamma_\kappa(p)-T\frac{\tilde{\gamma}+\frac{1}{2}}{\tilde{\beta}}\ln{\frac{\tilde{\beta}+\frac{1}{2}}{\tilde{\beta}-\frac{1}{2}}}=0\\
\rho_0 \widehat{V}(p) + \widehat{V}\ast \alpha_\kappa(p)+T\frac{\tilde{\alpha}}{\tilde{\beta}}\ln{\frac{\tilde{\beta}+\frac{1}{2}}{\tilde{\beta}-\frac{1}{2}}}=0 
\end{aligned}
\end{eqnarray}
is satisfied for $|p|>p_\kappa-\epsilon$ for some $\epsilon>0$. Of course we know that $(\gamma_\kappa,\alpha_\kappa)$ fulfils this equation for $|p|>p_\kappa$, but we can do a little better. To see that such $\tilde{\gamma}$ and $\tilde{\alpha}$ exist, consider the explicit expression \eqref{gammaexp} for $\tilde{\gamma}$ in terms of $G$ and $A$ which follows from \eqref{wanted11} as before (and only depends on $\gamma_\kappa$ and $\alpha_\kappa$). In particular, we know that as long as $|p|\geq p_\kappa$ the denominator in \eqref{gammaexp} does not go to zero since the $\gamma_\kappa\leq \kappa/p^2$ for $|p|> p_\kappa$. Since $G$ and $A$ are continuous everywhere, we can infer that there has to be a small region $|p|>p_\kappa-\epsilon$ where there exist continuous $\tilde{\gamma}$ and $\tilde{\alpha}$ that satisfy \eqref{wanted11} (which have to coincide with $\gamma_\kappa$ for $|p|> p_\kappa$, but may not do so otherwise). 

Now suppose that $\tilde\gamma(p)<\kappa/p_\kappa^2$ for $|p|=p_\kappa$. 
By continuity and the argument above, we then also
have that $\tilde{\gamma}(p)<\kappa/p^2$ on (a possibly smaller
region) $p_\kappa-\epsilon<|p|\leq p_\kappa$. By the definition of $p_\kappa$ we must then have that 
$\gamma_\kappa(p)>\tilde\gamma(p)$ on a set of positive measure. Since the
entropy derivative is strictly increasing in $\gamma$, we have for such $p$ that
\[
  \frac{\partial\mathcal{F}^{\rm{can}}_{\delta}}{\partial\gamma_\kappa}=p^2-\delta+\rho\widehat{V}(0)+\rho_0 \widehat{V}(p) + \widehat{V}\ast \gamma_\kappa(p)-T\frac{\gamma_\kappa+\frac{1}{2}}{\beta_k}\ln{\frac{\beta_k+\frac{1}{2}}{\beta_k-\frac{1}{2}}}>0.
\]
This contradicts the fact that 
$\gamma_\kappa$ is part of a minimizer (which should always satisfy
\eqref{eq:ELrestrictegamma}). We conclude that \eqref{gammaatpk} is
true.
\qed
\end{proof}

This lemma implies that we can pick a subsequence $p_\kappa$ that is decreasing and tends to zero. This is what we will assume from now on.

We now show that the corresponding $(\gamma_\kappa,\alpha_\kappa)$ form a minimizing sequence of \eqref{def:dualauxproblem}.

\begin{lemma}
Let $(\gamma_k,\alpha_k,\rho_0)$ be minimizers for the restricted problem \eqref{def:restricteddual}. We then have
\begin{equation*}
\lim_{\kappa\rightarrow\infty}\mathcal{F}^{\rm{can}}_{\delta}(\gamma_\kappa,\alpha_\kappa,\rho_0)=\inf_{(\gamma,\alpha)\in\mathcal{D}'}\mathcal{F}^{\rm{can}}_{\delta}(\gamma,\alpha,\rho_0).
\end{equation*}
\label{lem2}
\end{lemma}
\begin{proof}
  Let $(\gamma,\alpha)$ be a general element in $\mathcal{D}'$. We
  will show that its energy can always be approximated by the energy
  of a sequence of elements in $\mathcal{M}_\kappa$. We simply
  define the functions:
\begin{align*}
\tilde{\gamma}_\kappa&=\gamma\ \mathbbm{1}(\gamma\leq \kappa/p^2)\\
\tilde{\alpha}_\kappa&=\alpha\ \mathbbm{1}(\gamma\leq \kappa/p^2),
\end{align*}
which implies $(\tilde{\gamma}_\kappa,\tilde{\alpha}_\kappa)\in\mathcal{M}_\kappa$.
It follows from Lebesgue's Dominated Convergence Theorem that 
\begin{equation*}
\mathcal{F}^{\rm{can}}_{\delta}(\tilde{\gamma}_\kappa,\tilde{\alpha}_\kappa,\rho_0)\to\mathcal{F}^{\rm{can}}_{\delta}(\gamma,\alpha,\rho_0).
\end{equation*}
Since the $(\gamma_\kappa,\alpha_\kappa)$ are minimizers for the restricted problems,
\[
\mathcal{F}^{\rm{can}}_{\delta}(\gamma_\kappa,\alpha_\kappa,\rho_0)\leq\mathcal{F}^{\rm{can}}_{\delta}(\tilde{\gamma}_\kappa,\tilde{\alpha}_\kappa,\rho_0).
\]
By taking a limit in $\kappa$ followed by an infimum over $\mathcal{D}'$, we obtain 
\[
\limsup_{\kappa\to\infty}\mathcal{F}^{\rm{can}}_{\delta}(\gamma_\kappa,\alpha_\kappa,\rho_0)\leq\inf_{\mathcal{D}'}\mathcal{F}^{\rm{can}}_{\delta}(\gamma,\alpha,\rho_0),
\]
which in combination with the easy observation (use $\mathcal{M}_\kappa\subset\mathcal{D}'$ to get the inequality and then take the $\liminf$)
\[
\inf_{\mathcal{D}'}\mathcal{F}^{\rm{can}}_{\delta}(\gamma,\alpha,\rho_0)\leq\liminf_{\kappa\to\infty}\mathcal{F}^{\rm{can}}_{\delta}(\gamma_\kappa,\alpha_\kappa,\rho_0)
\]
leads to the desired conclusion.
\qed
\end{proof}

We will now construct a candidate minimizer $(\tilde{\gamma},\tilde{\alpha})$ from the restricted minimizers  $(\gamma_\kappa,\alpha_\kappa)$. Now it will be important that we are dealing with $(\delta^{\text{min}},\rho_0^{\text{min}})$ and not just any $(\delta,\rho_0)$.

\begin{proposition}
  Let $(\gamma_\kappa,\alpha_\kappa)$ be a sequence of minimizers for
  $\widehat{f}_\kappa(\delta^{\rm{min}},\rho_0^{\rm{min}})$
  such that $p_\kappa$ decreases to zero. We can extract a new
  minimizing sequence, denoted also $(\gamma_\kappa,\alpha_\kappa)$,
  such that $\gamma_\kappa\to\tilde{\gamma}$ pointwise and in $L^1$,
  and $\alpha_\kappa\to\tilde{\alpha}$ pointwise and
  $(\tilde{\gamma},\tilde{\alpha})\in\mathcal{D}'$.
\label{prop:minimizingsequenceunrestricted}
\end{proposition}
\begin{proof}
  \textit{Step 1 - pointwise convergence.} Recall definition
  \eqref{def:pkappa}. For a given $p\in \R^3$ define $\kappa_0(p)$ to
  be the smallest $\kappa$ such that $p_\kappa<|p|$. Let
  $h(p):=\kappa_0(p)/p^2$. Let us call $l(p)$ the function that
  defines the a priori large $p$ upper bound on $\gamma(p)$, i.e.\
  $l(p)$ is the RHS of \eqref{eq:largepboundgamma}. Similarly, let
  $s(p)$ be the function that defines the a priori small $p$ upper
  bound on $\gamma(p)$, i.e.\ $s(p)$ is the RHS of
  \eqref{eq:smallpboundgamma}. Note that using the fact that
  $\|\gamma_\kappa\|_1$ is bounded uniformly in $\kappa$, the bound
  $s(p)$ can be modified to be $\kappa$-independent. We call this new
  bound $s(p)$ as well. We then define the function
$$K(p)=\max \{h(p),l(p),s(p)\}.$$
Note that for any $\kappa$ we have
$$\gamma_\kappa(p)\leq K(p).$$
Indeed, given a $\gamma_\kappa$ and $p$ we either have $\kappa\leq \kappa_0(p)$ or $\kappa>\kappa_0(p)$. In the first case we clearly have $\gamma_\kappa (p)\leq \kappa/p^2\leq h(p)$. In the second case, by definition, we have $p_\kappa \leq |p|$ and thus $\gamma_\kappa (p)\leq  \max \{l(p),s(p)\}$.

We use the function $K(p)$ to introduce the weighted $L^2$-space with the measure $d\mu(p)=\frac{g(p)dp}{(K(p))^2}$ where $g(p)$ is a strictly positive $L^1$-function such that the measure is finite ($g$ has to decay sufficiently fast). We then have the uniform bounds
\bq
\begin{aligned}
\|\gamma_\kappa\|^2_{L^{2}(d\mu(p))}=&\int\frac{\gamma^2_\kappa g}{K^2}\leq \int g < C, \\
\|\alpha_\kappa\|^2_{L^{2}(d\mu(p))}=&\int \frac{\alpha^2_\kappa g}{K^2}\leq \int \frac{(\gamma^2_\kappa+\gamma_\kappa)g}{K^2} \leq C+\frac12\int\frac{\gamma^2_\kappa g}{K^2}+\frac12 \int\frac{g}{K^2} <C.
\end{aligned} \nn \eq These bounds allow us to extract a subsequence
$(\gamma_\kappa,\alpha_\kappa)$ that converges weakly in the weighted
$L^2$-space. Next, applying Mazur's Lemma, we can obtain a strongly
converging sequence of convex combinations, which -- by convexity of
the functional (recall Lemma \ref{lem:auxfunctionalconvexity}) -- is
also a minimizing sequence. Picking a further subsequence we can
obtain a pointwise converging subsequence. We denote the limiting
functions by $\tilde{\gamma}$ and $\tilde{\alpha}$. By the pointwise
convergence we have
$\tilde{\alpha}^2\leq\tilde{\gamma}(\tilde{\gamma}+1)$.

\textit{}\\
\textit{Step 2.} Fatou's lemma in combination with pointwise convergence implies that
\begin{equation}
\int (1+p^2)\tilde{\gamma}=\int \liminf_{\kappa\to\infty} (1+p^2)\gamma_\kappa\leq\liminf_{\kappa\to\infty} \int  (1+p^2)\gamma_\kappa.
\label{Fatou9}
\end{equation}
Recall that we have a uniform bound on
$\|\gamma_\kappa\|_{L^1((1+p^2)dp)}$.  This means that the integral on
the left-hand side is bounded and
therefore $\tilde{\gamma}\in L^1((1+p^2)dp)$.\\
\textit{}\\
\textit{Step 3 - $L^1$-convergence.} 
By Lemma \ref{lem2} we know that $(\gamma_\kappa,\alpha_\kappa)$ form a minimizing sequence for $\widehat{f}(\delta^{\rm{min}},\rho_0^{\rm{min}})$ and hence, by strict convexity of $f(\lambda,\rho_0)$ in $\lambda$, $\int \gamma_\kappa\to \lambda^{\text{min}}:=\rho-\rho^{\rm min}_0$.

Recall that the $\gamma_\kappa$ are uniformly bounded by an $L^1$ function on intervals $[\epsilon,\infty)$. This implies
\begin{equation}
\int_{|p|>\epsilon}\gamma_\kappa\xrightarrow{\kappa\to\infty}\int_{|p|>\epsilon}\tilde{\gamma}\xrightarrow{\epsilon\to0}\int\tilde{\gamma}
\label{someobs1}
\end{equation}
where the first convergence follows by an application of the Dominated Convergence Theorem (we have pointwise convergence and a uniform $L^1$-bound by \eqref{eq:largepboundgamma} and  \eqref{eq:smallpboundgamma}), and the second by the Monotone Convergence Theorem. Furthermore, it follows from Fatou's lemma that $\int\tilde{\gamma}\leq \lambda^{\text{min}}$. First assume that $\int\tilde{\gamma}= \lambda^{\text{min}}$.  We use this to see that
\begin{align*}
\int|\tilde{\gamma}-\gamma_\kappa|&\leq\int_{|p|>\epsilon}|\tilde{\gamma}-\gamma_\kappa|+\int_{|p|\leq\epsilon}\tilde{\gamma}+\int_{|p|\leq\epsilon}\gamma_\kappa\\
&\leq\int_{|p|>\epsilon}|\tilde{\gamma}-\gamma_\kappa|+\int_{|p|\leq\epsilon}\tilde{\gamma}+\int\gamma_\kappa-\int_{|p|>\epsilon}\gamma_\kappa.
\end{align*}
We use our observation \eqref{someobs1} for the fourth term and apply
the Dominated Convergence Theorem to the first term to obtain
\[
  \limsup_{\kappa\to\infty}\int|\tilde{\gamma}-\gamma_\kappa|\leq \int_{|p|\leq\epsilon}\tilde{\gamma}
  +\lambda^{\text{min}}-\int_{|p|>\varepsilon}\tilde{\gamma}\xrightarrow{\epsilon\to0}0.
\]
Here, the last terms cancel because of our assumption on $\int\tilde{\gamma}$ and the convergence in $\epsilon$ holds simply because $\tilde{\gamma}\in L^1$. This means that in this case we have proved the proposition. It remains to show that $\int\tilde{\gamma}< \lambda^{\text{min}}$ is impossible. \\
\textit{}\\
\textit{Step 4.} Assume $\int\tilde{\gamma}< \lambda^{\text{min}}$. We
have
\begin{equation*}
\int_{|p|\leq\epsilon}\gamma_\kappa=\int\gamma_\kappa-\int_{|p|>\epsilon}\gamma_\kappa\xrightarrow{\kappa\to\infty}\lambda^{\text{min}}-\int_{|p|>\epsilon}\tilde{\gamma}\xrightarrow{\epsilon\to0}\lambda^{\text{min}}-\int\tilde{\gamma}>0.
\end{equation*}
This quantity is important for our proof, so we give it a name:
\begin{equation*}
d:=\lim_{\epsilon\to 0}\lim_{\kappa\to\infty}\int_{|p|\leq\epsilon}\gamma_\kappa>0.
\end{equation*}

Our goal \eqref{conc123}, and eventually \eqref{conc193}, is to show that we can lower the energy by adding this mass to the $\rho_0$, which contradicts the fact that $\rho^{\rm min}_0$ is the minimizing point obtained in Proposition \ref{lem:auxfunctionalconvexity}. Hence, $\int\tilde{\gamma}< \lambda^{\text{min}}$ cannot occur. 

We start with the following estimate (throwing out some positive terms, using \eqref{gammaplusalpha} and estimating the entropy for small $p$):
\begin{eqnarray}
\begin{aligned}
\mathcal{F}^{\text{can}}(\gamma_\kappa,&\alpha_\kappa,\rho_0^{\text{min}})\geq\int_{|p|>\epsilon} p^2\gamma_\kappa+\rho_0^{\text{min}}\int_{|p|>\epsilon}\widehat{V}(\gamma_\kappa+\alpha_\kappa)-\frac{1}{2}\rho_0^{\text{min}}\int_{|p|\leq\epsilon}\widehat{V}\\
&-T\int_{|p|>\epsilon}s(\gamma_\kappa,\alpha_\kappa)-CT\epsilon^3-CT\|\gamma_\kappa\|_1^{\frac12}\epsilon^{3/2} \\
&+\frac12\widehat{V}(0)\left(\int\gamma_\kappa+\rho^{\rm min}_0\right)^2\\
&+\frac{1}{2}\int_{|p|>\epsilon}\int_{|q|>\epsilon} \widehat{V}(p-q)(\gamma_\kappa(p)\gamma_\kappa(q)+\alpha_\kappa(p)\alpha_\kappa(q))dqdp\\
&+\int_{|p|\leq\epsilon}\int_{|q|>\epsilon} \widehat{V}(q-p)(\gamma_\kappa(p)\gamma_\kappa(q)+\alpha_\kappa(p)\alpha_\kappa(q))dqdp\\
&+\frac{1}{2}\int_{|p|\leq\epsilon}\int_{|q|\leq\epsilon} \widehat{V}(p-q)\gamma_\kappa(p)\gamma_\kappa(q)dqdp.
\label{intermstep11}
\end{aligned}
\end{eqnarray}
Note that we have obtained the term in the fourth line twice since $\widehat{V}$ is radial, which implies $\widehat{V}(p-q)=\widehat{V}(q-p)$.
For the entropy, we have used
\begin{equation}
\int_{|p|<\epsilon}s(\gamma,\alpha)\leq \int_{|p|<\epsilon}s(\gamma,0)=\int_{|p|<\epsilon}(1+\gamma)\ln(1+\gamma)-\gamma\ln\gamma.
\label{someref22}
\end{equation}
In the region where $\gamma\leq1$, the integrand is bounded by $2\ln(2)+1$. In the region where $\gamma>1$, we have 
\begin{equation*}
(1+\gamma)\ln(1+\gamma)-\gamma\ln\gamma=\ln\gamma+(1+\gamma)\ln(1+\gamma^{-1})\leq\ln\gamma+1+\gamma^{-1}\leq\sqrt{\gamma}+2.
\end{equation*}
Together with \eqref{someref22}, using Cauchy--Schwarz, this implies that
\begin{equation*}
\int_{|p|<\epsilon}s(\gamma,\alpha)\leq C\epsilon^3+C\|\gamma_\kappa\|_1^{\frac12}\epsilon^{3/2}.
\end{equation*}
Continuing from \eqref{intermstep11}, for $|p|\leq\epsilon$ we estimate 
\begin{equation}
\left|\int_{|q|>\epsilon}\widehat{V}(q)\gamma_\kappa(q)dq-\int_{|q|>\epsilon}\widehat{V}(q-p)\gamma_\kappa(q)dq\right|\leq
\epsilon\|\nabla\widehat{V}\|_\infty\|\gamma_\kappa\|_1,
\label{someest00}
\end{equation}
where we have used our assumptions on the differentiability of $\widehat{V}$. 
To see that a similar estimate holds for $\alpha_\kappa$, we note that  by an argument identical to \eqref{convValpha} we have
\begin{equation*}
\left\|\nabla\int_{|q|>\epsilon}\widehat{V}(q-p)\alpha_\kappa(q)dq\right\|_\infty=\left\|\int_{|q|>\epsilon}\nabla\widehat{V}(q-p)\alpha_\kappa(q)dq\right\|_\infty\leq C
\end{equation*}
where $C$ is a constant that can be chosen independent of $\kappa$. For $|p|\leq\epsilon$ this leads to
\begin{equation*}
\left|\int_{|q|>\epsilon}\widehat{V}(q)\alpha_\kappa(q)-\int_{|q|>\epsilon}\widehat{V}(q-p)\alpha_\kappa(q)\right|
\leq\epsilon C.
\end{equation*}
Finally, for $|q|\leq\epsilon$, 
\begin{equation*}
\left|\int_{|p|\leq\epsilon}\widehat{V}(0)\gamma_\kappa(p)-\int_{|p|\leq\epsilon}\widehat{V}(p-q)\gamma_\kappa(p)\right|
\leq 2\epsilon\|\nabla\widehat{V}\|_\infty\|\gamma_\kappa\|_1.
\end{equation*}
Using the last two estimates together with \eqref{someest00} in \eqref{intermstep11} and estimating the third term of \eqref{intermstep11}, we obtain
\begin{eqnarray*}
\begin{aligned}
\mathcal{F}^{\text{can}}(\gamma_\kappa,&\alpha_\kappa,\rho_0^{\text{min}})\geq\int_{|p|>\epsilon} p^2\gamma_\kappa+\rho_0^{\text{min}}\int_{|p|>\epsilon}\widehat{V}(\gamma_\kappa+\alpha_\kappa)-T\int_{|p|>\epsilon}s(\gamma_\kappa,\alpha_\kappa)\\
&+\frac{1}{2}\int_{|p|>\epsilon}\int_{|q|>\epsilon} \widehat{V}(p-q)(\gamma_\kappa(p)\gamma_\kappa(q)+\alpha_\kappa(p)\alpha_\kappa(q))dqdp\\
&+\frac12\widehat{V}(0)\left(\int\gamma_\kappa+\rho^{\rm min}_0\right)^2\\
&+\left(\int_{|p|\leq\epsilon}\gamma_\kappa(p)dp\right)\left[\int_{|q|>\epsilon} \widehat{V}(q)\gamma_\kappa(q)dq-\epsilon\|\nabla\widehat{V}\|_\infty\|\gamma_\kappa\|_1\right]\\
&-\left|\int_{|p|\leq\epsilon}\alpha_\kappa(p)dp\right|\left[\left|\int_{|q|>\epsilon} \widehat{V}(q)\alpha_\kappa(q)dq\right|+\epsilon C\left(\|\gamma_\kappa\|_1,\nabla\widehat{V}\right)\right]\\
&+\frac{1}{2}\left(\int_{|q|\leq\epsilon}\gamma_\kappa(p)dp\right)\left[\widehat{V}(0)\int_{|p|\leq\epsilon}\gamma_\kappa(q)dq-2\epsilon\|\nabla\widehat{V}\|_\infty\|\gamma_\kappa\|_1\right]\\
&-C\rho_0^{\text{min}}\|\widehat{V}\|_\infty-CT\epsilon^3-CT\|\gamma_\kappa\|_1^{\frac12}\epsilon^{3/2}.
\end{aligned}
\end{eqnarray*}
Since $|\alpha_\kappa|\leq\gamma_\kappa+1/2$, we see that
\begin{equation}\label{eq:alphabound}
\left|\int_{|p|\leq\epsilon}\alpha_\kappa(p)dp\right|\leq\int_{|p|\leq\epsilon}\gamma_\kappa(p)dp+C\epsilon^3.
\end{equation}
and hence all the error terms in this expression tend to zero as
$\epsilon\to0$. 

We now choose $\kappa(\epsilon)$ such that it tends to infinity as $\epsilon\to0$ and such that 
$
\left|\lim_{\kappa\to\infty}\int_{|p|<\epsilon}\gamma_\kappa-\int_{|p|<\epsilon}\gamma_{\kappa(\epsilon)}\right|<\epsilon.
$
Then, in particular, 
\begin{equation}\label{eq:kappaeps}
\lim_{\epsilon\to0}\int_{|p|<\epsilon}\gamma_{\kappa(\epsilon)}-d\to0.
\end{equation}
Combining this with \eqref{eq:alphabound} we find that 
\begin{eqnarray*}
  \lefteqn{\mathcal{F}^{\text{can}}(\gamma_{\kappa(\epsilon)},\alpha_{\kappa(\epsilon)},\rho_0^{\text{min}})}&&\\
  &\geq&\int_{|p|>\epsilon} p^2\gamma_{\kappa(\epsilon)}
  -T\int_{|p|>\epsilon}s(\gamma_{\kappa(\epsilon)},\alpha_{\kappa(\epsilon)})\\
  &&+\rho_0^{\text{min}}\left(\int_{|p|>\epsilon}\widehat{V}\gamma_{\kappa(\epsilon)}
    -\left|\int_{|p|>\epsilon}\widehat{V}\alpha_{\kappa(\epsilon)}\right|\right)\\
  &&+\frac{1}{2}\int_{|p|>\epsilon}\int_{|q|>\epsilon} \widehat{V}(p-q)(\gamma_{\kappa(\epsilon)}(p)\gamma_{\kappa(\epsilon)}(q)+\alpha_{\kappa(\epsilon)}(p)\alpha_{\kappa(\epsilon)}(q))dqdp\\
&&+\frac12\widehat{V}(0)\left(\int_{|p|>\epsilon}\gamma_\kappa(\epsilon)(p)dp+\rho^{\rm min}_0+d\right)^2\\
  &&+d\int_{|p|>\epsilon}\widehat{V}\gamma_{\kappa(\epsilon)}
  -d\left|\int_{|p|>\epsilon}\widehat{V}\alpha_{\kappa(\epsilon)}\right|+\frac{1}{2}d^2\widehat{V}(0)-c(\epsilon),
\end{eqnarray*}
where the error $c(\epsilon)$ tends to zero as $\epsilon\to0$. We now define
\begin{align*}
\tilde{\gamma}_{\epsilon}&=\gamma_{\kappa(\epsilon)}\ \mathbbm{1}(|p|>\epsilon)\\
\tilde{\alpha}_{\epsilon}&=\pm\alpha_{\kappa(\epsilon)}\ \mathbbm{1}(|p|>\epsilon),
\end{align*}
where the sign $\pm$ in the second equation is chosen such that 
\[
\int_{|p|\geq\epsilon}\widehat{V}\tilde\alpha_\epsilon\leq0.
\]
Then 
\begin{equation}
\label{conc123}
  \mathcal{F}^{\text{can}}(\gamma_{\kappa(\epsilon)},\alpha_{\kappa(\epsilon)},\rho_0^{\text{min}})
  \geq\mathcal{F}^{\text{can}}(\tilde{\gamma}_{\epsilon},\tilde{\alpha}_{\epsilon},\rho_0^{\text{min}}+d)
  +\frac{1}{2}d^2\widehat{V}(0)-c(\epsilon).
\end{equation}
We will now show how \eqref{conc123} leads to a contradiction if
$d>0$. We first use that applying the Legendre transform twice on a convex
function yields the original function (recall that
$f(\lambda,\rho_0)$ is strictly convex in $\lambda$). Thus 
\begin{align*}
f(\lambda,\rho_0)
=\sup_{\delta}\left[\inf_{\lambda'}\left[f(\lambda',\rho_0)-\delta\lambda'\right]
+\delta\lambda\right],
\end{align*}
and hence
\begin{equation}
f(\lambda,\rho_0)=\sup_{\delta}\left[\inf_{(\gamma,\alpha)\in\mathcal{D}'}\left[\cF^{\rm can}(\gamma,\alpha,\rho_0)-\delta\int\gamma\right]+\delta\lambda\right].
\label{helpf1}
\end{equation}
Using Lemma \ref{lem2} (recall that $\cF^{\rm can}_\delta(\gamma,\alpha,\rho_0)
=\cF^{\rm can}(\gamma,\alpha,\rho_0)-\delta\int\gamma$) and the fact that $\lim_{\kappa\to\infty}\int\gamma_\kappa=\lambda^{\rm min}$, we note that for any $\delta\in\mathbb{R}$:
\[
\begin{aligned}
f(\lambda^{\text{min}},\rho_0^{\text{min}})&\geq\inf_{(\gamma,\alpha)\in\mathcal{D}'}\left[\cF^{\rm can}(\gamma,\alpha,\rho^{\rm min}_0)-\delta^{\rm min}\int\gamma\right]+\delta^{\rm min}\lambda^{\rm min}\\
&=\lim_{\kappa\to\infty}\left[\mathcal{F}^{\text{can}}(\gamma_\kappa,\alpha_\kappa,\rho_0^{\text{min}})
-\delta\int\gamma_\kappa\right]+\delta\lambda^{\text{min}}.
\end{aligned}
\]
Recalling our conclusion \eqref{conc123} and 
that \eqref{eq:kappaeps} gives 
$\lim_{\epsilon\to0}\int(\gamma_{\kappa(\epsilon)}-\tilde{\gamma}_{\epsilon})=d$, we obtain
\bq
\begin{aligned}
f(\lambda^{\text{min}},\rho_0^{\text{min}})&\geq
  \liminf_{\epsilon\to0}
  \left[\mathcal{F}^{\text{can}}(\tilde{\gamma}_{\epsilon},\tilde{\alpha}_{\epsilon},\rho_0^{\text{min}}+d)
    -\delta\left(\int\tilde{\gamma}_{\epsilon}+d\right)\right]\\
  &\quad+\frac{1}{2}d^2\widehat{V}(0)+\delta\lambda^{\text{min}}\\
  &\geq\inf_{(\gamma,\alpha)\in\mathcal{D}'}\left[\mathcal{F}^{\text{can}}(\gamma,\alpha,\rho_0^{\rm{min}}+d)
    -\delta\int\gamma\right]+\delta(\lambda^{\text{min}}-d)\\
  &\quad+\frac{1}{2}d^2\widehat{V}(0),
\end{aligned} 
\nn 
\eq 
where we have also used that
$(\tilde{\gamma}_{\epsilon},\tilde{\alpha}_{\epsilon})\in\mathcal{D}'$
for all $\epsilon$.  By taking a supremum over $\delta$ on both sides
and using \eqref{helpf1}, we obtain
\begin{equation}
\label{conc193}
f(\lambda^{\text{min}},\rho_0^{\text{min}})\geq f(\lambda^{\text{min}}-d,\rho_0^{\text{min}}+d)+\frac{1}{2}d^2\widehat{V}(0).
\end{equation}
Thus, if $d>0$ we arrive at a contradiction with the fact that
$(\lambda^{\text{min}},\rho_0^{\text{min}})=(\rho-\rho_0^{\text{min}},\rho_0^{\text{min}})$ is the minimum of
$f(\rho-\rho_0,\rho_0)$ since $\widehat{V}(0)>0$ . This means the case
$\int\gamma<\lambda^{\text{min}}$ cannot occur. Since we had already
proved the claims for the other case, this concludes the proof of the
proposition.
\qed
\end{proof}

We need a final lemma to show the existence of a minimizer for the dual problem \eqref{def:dualauxproblem} at the relevant $\delta^{\rm min}$ and $\rho_0^{\rm min}$.
\begin{lemma}
Let $(\gamma_\kappa,\alpha_\kappa)$ and $(\tilde{\gamma},\tilde{\alpha})$ be as above. In particular, we have $\gamma_\kappa\to\tilde{\gamma}$ pointwise and in $L^1$, and $\alpha_\kappa\to\tilde{\alpha}$ pointwise. We then have 
\begin{equation*}
\liminf_{\kappa\to\infty}\mathcal{F}^{\rm{can}}_{\delta^{\rm{min}}}(\gamma_\kappa,\alpha_\kappa,\rho_0^{\rm{min}})\geq\mathcal{F}^{\rm{can}}_{\delta^{\rm{min}}}(\tilde{\gamma},\tilde{\alpha},\rho_0^{\rm{min}}).
\end{equation*}
\label{lem1}
\end{lemma}
\begin{proof}
We recall that
\begin{equation}
\label{somefunc}
\begin{aligned}
\mathcal{F}^{\rm{can}}_{\delta^{\text{min}}}(\gamma,\alpha,\rho_0^{\text{min}})&=\int p^2\gamma(p)dp-T\int s(\gamma(p),\alpha(p))dp -\delta^{\text{min}}\int\gamma(p)dp \\
&+\frac12\widehat{V}(0)(\int\gamma(p)dp+\rho^{\text{min}}_0)^2+\rho_0^{\text{min}}\int\widehat{V}(p)(\gamma(p)+\alpha(p))dp \\
&+\frac{1}{2}\int\int \widehat{V}(p-q)(\gamma(p)\gamma(q)+\alpha(p)\alpha(q))dpdq.
\end{aligned}
\end{equation}
The third and fourth terms on the right-hand side simply converges because of the $L^1$-convergence of the $\gamma_\kappa$.
The combination of the first two terms is bounded below by an integrable function (as in \eqref{noninteractingbound}) and thus we can use pointwise convergence in combination with Fatou's lemma to conclude
\begin{equation*}
\int p^2\tilde{\gamma}-T\int s(\tilde{\gamma},\tilde{\alpha})
\leq\liminf_{\kappa\to\infty}\left(\int p^2\gamma_\kappa-T\int s(\gamma_\kappa,\alpha_\kappa)\right).
\end{equation*}
To show that the fourth term in \eqref{somefunc} also converges, we use two estimates. The easier one is 
\begin{equation}
\left|\rho_0^{\text{min}}\int\widehat{V}\left(\gamma_\kappa-\tilde{\gamma}\right)\right|\leq \rho_0^{\text{min}}\|\widehat{V}\|_\infty\int\left|\gamma_\kappa-\tilde{\gamma}\right|,
\label{somref121}
\end{equation}
which goes to zero by the  $L^1$-convergence of the $\gamma_\kappa$.
 For the term involving $\tilde{\alpha}$, we write for $\epsilon>0$
\begin{align*}
\int_{|p|\leq\epsilon}|\alpha_\kappa|&=\int_{|p|\leq\epsilon,|\gamma_\kappa|\leq1}|\alpha_\kappa|+\int_{|p|\leq\epsilon,|\gamma_\kappa|>1}|\alpha_\kappa|
&\leq C\epsilon^3+\sqrt{2}\int_{|p|\leq\epsilon}\gamma_\kappa,
\end{align*}
where we have used the usual estimate on $\alpha_\kappa$ in terms of $\gamma_\kappa$. Note that this also holds for $\tilde{\alpha}$ (in terms of $\tilde{\gamma}$). 
For $|p|>\epsilon$ and $\kappa$ large enough we see from 
$|\alpha_\kappa|^2\leq\gamma_\kappa(\gamma_\kappa+1)$ and Lemmas \ref{lem:boundgammalargep} and 
\ref{lem:boundgammasmallp} that the Dominated Convergence Theorem gives
$$
\lim_{\kappa\to\infty}\int_{|p|>\epsilon}|\tilde{\alpha}-\alpha_\kappa|^2=0.
$$
Hence
\begin{eqnarray}
\begin{aligned}
&\int|\widehat{V}||\tilde{\alpha}-\alpha_\kappa|\leq\|\widehat{V}\|_\infty\int_{|p|\leq\epsilon}|\tilde{\alpha}-\alpha_\kappa|+\int_{|p|>\epsilon}|\widehat{V}||\tilde{\alpha}-\alpha_\kappa|\\
&\leq \|\widehat{V}\|_\infty\int_{|p|\leq\epsilon}\left(|\tilde{\alpha}|+|\alpha_\kappa|\right)+\left(\int_{|p|>\epsilon}|\widehat{V}|^2\right)^{\frac12}\left(\int_{|p|>\epsilon}|\tilde{\alpha}-\alpha_\kappa|^2\right)^{\frac12}\\
&\leq C \|\widehat{V}\|_\infty\epsilon^3+\sqrt{2}\|\widehat{V}\|_\infty\int_{|p|\leq\epsilon}\left(\gamma_\kappa+\tilde{\gamma}\right)+C\left(\int_{|p|>\epsilon}|\tilde{\alpha}-\alpha_\kappa|^2\right)^{\frac12}\\
&\xrightarrow{\kappa\to\infty}\ \ C\|\widehat{V}\|_\infty\epsilon^3+2\sqrt{2}\|\widehat{V}\|_\infty\int_{|p|\leq\epsilon}\tilde{\gamma}. \label{alphaL1boundconvergence}
\end{aligned}
\end{eqnarray}
Since this holds for any $\epsilon>0$ and tends to 0 as $\epsilon\to0$, we combine our conclusion with the first estimate to see that the entire third term converges, i.e.\
\begin{equation*}
\rho_0^{\text{min}}\int\widehat{V}(p)\left(\gamma_\kappa(p)+\alpha_\kappa(p)\right)dp\to\rho_0^{\text{min}}\int\widehat{V}(p)\left(\tilde{\gamma}(p)+\tilde{\alpha}(p)\right)dp.
\end{equation*}
Finally, we need to take care of the fifth term in \eqref{somefunc}.
It is enough to bound
\begin{eqnarray}
  \lefteqn{\left|\sqrt{\iint \tilde{\gamma}(p)\widehat{V}(p-q)\tilde{\gamma}(q)dpdq}
      -\sqrt{\iint\gamma_\kappa(p)\widehat{V}(p-q)\gamma_\kappa(q)dpdq}\right|}&&
  \label{somref124}\\&\leq&\sqrt{\iint(\gamma_\kappa(p)-\tilde\gamma(p))\widehat{V}(p-q)(\gamma_\kappa(q)-\tilde\gamma(q))dpdq}
  \leq\|\widehat V\|_\infty^{1/2}\|\gamma_\kappa-\tilde\gamma\|_1\nn,
\end{eqnarray}
where we have used that $V\geq0$.  This implies convergence of the
$\gamma$-part of the fifth term. Since we do not have
$L^1$-convergence for $\alpha_\kappa$, we need to use a different
method. We again need to control
\begin{eqnarray*}
\lefteqn{\iint(\alpha_\kappa-\tilde{\alpha})(p)\widehat{V}(p-q)(\alpha_\kappa-\tilde{\alpha})(q)dpdq}&&\\
&\leq& 2\iint_{|p|,|q|>\epsilon}(\alpha_\kappa-\tilde{\alpha})(p)\widehat{V}(p-q)(\alpha_\kappa-\tilde{\alpha})(q)dpdq
\\&&+2\iint_{|p|,|q|<\epsilon}(\alpha_\kappa-\tilde{\alpha})(p)\widehat{V}(p-q)(\alpha_\kappa-\tilde{\alpha})(q)dpdq.
\end{eqnarray*}
For the first integral we use
\begin{eqnarray}
\begin{aligned}
\int_{|p|,|q|>\epsilon}&(\alpha_\kappa-\tilde{\alpha})(p)\widehat{V}(p-q)(\alpha_\kappa-\tilde{\alpha})(q)dpdq=\\ 
&=\int V(x)|\mathcal{F}^{-1}\left((\alpha_\kappa-\tilde{\alpha})\mathbbm{1}(|p|>\epsilon)\right)|^2 dx\\
&\leq  \|V\|_{\infty}\|\mathcal{F}^{-1}\left(\left(\alpha_\kappa-\tilde{\alpha}\right)\mathbbm{1}(|p|>\epsilon)\right)\|_2^2=\|V\|_{\infty}\int_{|p|>\epsilon}|\alpha_\kappa-\tilde{\alpha}|^2, \label{case1}
\end{aligned}
\end{eqnarray}
where $\mathcal{F}^{-1}$ denotes the inverse Fourier transform. The second integral is bounded by
\begin{eqnarray}
\begin{aligned}
\int\limits_{|p|,|q|\leq\epsilon}(\alpha_\kappa-\tilde{\alpha})&(p)\widehat{V}(p-q)(\alpha_\kappa-\tilde{\alpha})(q)dpdq
\leq \|V\|_1\left(\ \int\limits_{|p|\leq\epsilon}|\alpha_\kappa-\tilde{\alpha}|dp \right)^2 \\ 
& \quad\quad\leq  \|V\|_1\left(C\epsilon^3+C\int_{|p|\leq\epsilon}\left(\gamma_\kappa+\tilde{\gamma}\right)\right)^2,  \label{case2}
\end{aligned}
\end{eqnarray}
where we used the same bound as in the first term of
\eqref{alphaL1boundconvergence}. Taking the limit $\kappa\to\infty$
followed by $\epsilon\to0$ in \eqref{case1}, \eqref{case2} and the
bound above, we see that we have convergence of the $\alpha$-part of
the fifth term. This concludes the proof of the lemma.
\qed
\end{proof}

We are ready to prove the main statement of this subsection, and hence Theorem \ref{thm:existencecanonicalepositiveT}. 

\begin{proof}[Proof of Proposition \ref{thm:existenceunconstrainedminimizers}]
We combine the previous two lemmas to obtain
\begin{eqnarray*}
\begin{aligned}
\inf_{(\gamma,\alpha)\in\mathcal{D}'}\mathcal{F}^{\rm{can}}_{\delta^{\text{min}}}(\gamma,\alpha,\rho_0^{\text{min}})&=\liminf_{\kappa\to\infty}\mathcal{F}^{\rm{can}}_{\delta^{\text{min}}}(\gamma_\kappa,\alpha_\kappa,\rho_0^{\text{min}})\geq\mathcal{F}^{\rm{can}}_{\delta^{\text{min}}}(\tilde{\gamma},\tilde{\alpha},\rho_0^{\text{min}}) \\ 
& \geq\inf_{(\gamma,\alpha)\in\mathcal{D}'}\mathcal{F}^{\rm{can}}_{\delta^{\text{min}}}(\gamma,\alpha,\rho_0^{\text{min}}),
\end{aligned}
\end{eqnarray*}
where the first equality holds by Lemma \ref{lem2} and the first inequality holds by Lemma \ref{lem1}.
We conclude that the $(\tilde{\gamma},\tilde{\alpha})$ constructed in Proposition \ref{prop:minimizingsequenceunrestricted} is a minimizer.
\qed
\end{proof}

\begin{remark}
How should the arguments above be adapted for the grand canonical functional \eqref{def:grandcanfreeenergyfunctional} and Theorem \ref{thm:existencepositiveT}?
By the definitions of $F(T,\mu)$ \eqref{def:grandcanonicalminimization} and $f(\lambda,\rho_0)$ \eqref{def:auxminimization}, we have 
\bq
\begin{aligned}
  F(T,\mu)=
  \inf_{\lambda,\rho_0}\left[f(\lambda,\rho_0)
    -\mu(\lambda+\rho_0)\right]. \label{eq:infauxformulation}
\end{aligned}
\eq
Using \eqref{gammaplusalpha} and \eqref{noninteractingbound}, we see that
\begin{align*}
f(\lambda,\rho_0)\geq -C_1 \rho_0 - C_2+\frac12\widehat{V}(0)(\lambda+\rho_0)^2,
\end{align*}
so that by continuity of $f$, in analogy with Proposition \ref{lem:auxfunctionalconvexity}, the infimum in \eqref{eq:infauxformulation} is attained at some point $(\lambda^{\rm min},\rho^{\rm min}_0)$, where $\lambda^{\rm min}$ and $\rho^{\rm min}_0$ are now independent. All arguments now go through as before to obtain a minimizer $(\gamma,\alpha)$ for $f(\lambda^{\rm min},\rho^{\rm min}_0)-\mu(\lambda^{\rm min}+\rho_0^{\rm min})$. Crucially, \eqref{conc193} still leads to a contradiction since $-\mu(\lambda+\rho_0)$ only depends on the sum of $\lambda$ and $\rho_0$. 
\end{remark}

\begin{remark}
The statement remains true in one and two dimensions.
\end{remark}

\section{Existence of minimizers for $T=0$}\label{sec:exis_min_proof_T=0}

In this section, we prove Theorems \ref{thm:existencezeroT} and \ref{thm:existencecanonicalzeroT}. The proof of the existence of minimizers for $T>0$ relied upon the bounds derived in Section \ref{apgr}. These showed that the minimizers of the restricted problem are uniformly bounded for fixed $T$, which allowed us to extract a limit. However, the bound deteriorates as $T\rightarrow0$ and hence the proof cannot be used for $T=0$. In this section we prove the existence of a minimizer for $T=0$ in a different way. 

\subsection{The grand canonical case}
\label{subsect:gcT=0} 
We first consider the grand canonical functional.
Note that the statement is trivial for $T=0$ and $\mu\leq 0$, since in this case the functional is obtained by taking expectation values of a positive operator. The minimizer is given by the vacuum, i.e.\ $(\gamma,\alpha,\rho_0)=(0,0,0)$. 

The rest of this subsection is dedicated to proving the theorem for $\mu>0$.
By the main result of the previous section, we know that for any $\mu$ and $T>0$ there exists a minimizer of the grand canonical functional \eqref{def:grandcanfreeenergyfunctional}. In this section, we will denote this functional by $\mathcal{F}^T$ to make the $T$-dependence explicit. As the proposition below shows, its minimizers at temperature $T$ actually form a minimizing sequence as $T\to0$ for the $T=0$ case.
\begin{proposition}[$T=0$ minimizing sequence]\label{prop:T=0minseq}
Let $(\gamma^T,\alpha^T,\rho^T_0)$ be a minimizer for $F(T,\mu)$ with $\mu,T>0$. Then
\bq
TS(\gamma^T,\alpha^T)\xrightarrow{T\to 0} 0 \quad \text{and}\quad \mathcal{F}^0(\gamma^T,\alpha^T,\rho_0^T)\xrightarrow{T\to 0} F(0,\mu).\nn
\eq
\end{proposition}
\begin{proof}
Let $T_1<T_2$. Making use of the minimizers at these temperatures, we obtain
\bq
\begin{aligned}
\mathcal{F}^{T_1}(\gamma^{T_1},\alpha^{T_1},\rho_0^{T_1})&=\mathcal{F}^{0}(\gamma^{T_1},\alpha^{T_1},\rho_0^{T_1})-T_1S(\gamma^{T_1},\alpha^{T_1})\\
&\leq\mathcal{F}^{0}(\gamma^{T_2},\alpha^{T_2},\rho_0^{T_2})-T_1S(\gamma^{T_2},\alpha^{T_2})\\
&=\mathcal{F}^{0}(\gamma^{T_2},\alpha^{T_2},\rho_0^{T_2})-T_2S(\gamma^{T_2},\alpha^{T_2})+(T_2-T_1)S(\gamma^{T_2},\alpha^{T_2})\\
&\leq\mathcal{F}^{0}(\gamma^{T_1},\alpha^{T_1},\rho_0^{T_1})-T_2S(\gamma^{T_1},\alpha^{T_1})+(T_2-T_1)S(\gamma^{T_2},\alpha^{T_2}). \nn
\end{aligned}
\eq
Comparing the first and last line we see that $S(\gamma^{T_1},\alpha^{T_1})\leq S(\gamma^{T_2},\alpha^{T_2})$, and thus the entropy of the minimizers decreases when $T$ does. Since it has to be non-negative, this implies that $TS(\gamma^T,\alpha^T)\to0$ as $T\to0$. 

Now, note that for all $(\gamma,\alpha,\rho_0)\in\mathcal{D}$ one has
\bq
\mathcal{F}^0(\gamma^T,\alpha^T,\rho_0^T)-TS(\gamma^T,\alpha^T)\leq\mathcal{F}^0(\gamma,\alpha,\rho_0)-TS(\gamma,\alpha).\nn
\eq
Taking a $\limsup_{T\to0}$ followed by an infimum over $(\gamma,\alpha,\rho_0)$ and combining this with
\bq
\liminf_{T\to0}\mathcal{F}^0(\gamma^T,\alpha^T,\rho_0^T)\geq F(0,\mu), \nn
\eq
proves the second claim.
\qed
\end{proof}

Now that we know that $\left\{(\gamma^T,\alpha^T,\rho^T_0)\right\}_{\left\{T>0\right\}}$ is a minimizing sequence, we would like to extract a limit out of it. This can in fact be done.

\begin{proposition}
There exists a subsequence of $\left\{(\gamma^T,\alpha^T,\rho^T_0)\right\}_{\left\{T>0\right\}}$ such that $\gamma^T\to\tilde{\gamma}$ pointwise and in $L^1$, $\alpha^T\to\tilde{\alpha}$ pointwise and in $L^2$, and $\rho^T_0\to\tilde{\rho}_0$. Moreover, the limit is an admissible state, i.e.\ $(\tilde{\gamma},\tilde{\alpha},\tilde{\rho}_0)\in\mathcal{D}$.
\label{limit0}
\end{proposition}

We will first state the proof of Theorem \ref{thm:existencezeroT}. The rest of this section will then be dedicated to proving Proposition \ref{limit0}.
\begin{proof}[Proof of Theorem \ref{thm:existencezeroT}]
As mentioned at the beginning of this section, the functional with $\mu\leq0$ has a minimizer $\gamma=\alpha=\rho_0=0$, so there is nothing left to prove. We consider the case $\mu>0$. By Proposition \ref{limit0}, we can assume that a suitable subsequence of $(\gamma^T,\alpha^T,\rho^T_0)$ has the convergence properties stated. Let us recall what the relevant functional looks like:
\begin{align*}
\mathcal{F}^0(\gamma^T,\alpha^T,\rho^T_{0})= &\int p^{2}\gamma^T(p)dp-\mu\rho^T+\frac{1}{2}\widehat{V}(0)\left(\rho^T\right)^{2}\\ 
&+\rho^T_{0}\int\widehat{V}(p)(\gamma^T(p)+\alpha^T(p))dp\\
&+\frac{1}{2}\int\widehat{V}(p-q)\left[\alpha^T(p)\alpha^T(q)+\gamma^T(p)\gamma^T(q)\right]dpdq.
\end{align*} 
We will show that this converges to something that is bigger than or equal to $\mathcal{F}^0(\tilde{\gamma},\tilde{\alpha},\tilde{\rho}_0)$, much like in Lemma \ref{lem1}.
The first term can be treated by Fatou's lemma and pointwise convergence (see \eqref{Fatou9} for a similar application). The second and third terms simply converge since $\rho_0^T\to\tilde{\rho}_0$ and $\rho^T_{\gamma}\to\rho_{\tilde{\gamma}}$ by $L^1$-convergence. The remaining terms involving $\gamma^T$ converge because of $L^1$-convergence (see estimates \eqref{somref121} and \eqref{somref124}). The quadratic $\alpha^T$-term is taken care of using $L^2$-convergence and the estimate \eqref{case1}, where now the integrals are over all $p$ and $q$. $L^2$-convergence also suffices to show convergence of the term linear in $\alpha^T$:
\begin{equation*}
\left|\int\widehat{V}\tilde{\alpha}-\int\widehat{V}\alpha^T\right|\leq\int\widehat{V}|\tilde{\alpha}-\alpha^T|\leq \left(\int |\widehat{V}|^2\right)^{1/2}\left(\int \left|\tilde{\alpha}-\alpha^T\right|^2\right)^{1/2}.
\end{equation*}
We have thus shown that
\begin{equation*}
\liminf_{T\rightarrow0}\mathcal{F}^0(\gamma^T,\alpha^T,\rho_0^T)\geq\mathcal{F}^0(\tilde{\gamma},\tilde{\alpha},\tilde{\rho}_0).
\end{equation*}
Together with Proposition \ref{prop:T=0minseq}, this leads to
\begin{equation*}
F(0,\mu)=\liminf_{T\rightarrow0}\mathcal{F}^0(\gamma^T,\alpha^T,\rho_0^T)\geq\mathcal{F}^0(\tilde{\gamma},\tilde{\alpha},\tilde{\rho}_0)\geq F(0,\mu),
\end{equation*}
which proves that $(\tilde{\gamma},\tilde{\alpha},\tilde{\rho}_0)$ is indeed a minimizer.
\qed
\end{proof}

It remains to prove Proposition \ref{limit0}. As mentioned before, some bounds in Section \ref{apgr} cannot be obtained uniformly in $T$, so they are useless for this case. However, the equivalent of Lemma \ref{lem:boundgammalargep} (with $\mu$ rather than $\delta$) does hold uniformly.
\begin{lemma}
Let $\mu\in\mathbb{R}$. There exist $C, P_0,T_0>0$ such that for all $|p|>P_0$ and $0<T\leq T_0$, we have
\[
\gamma^T(p)\leq C |p|^{-4}. 
\]
\label{largepboundgamma2}
\end{lemma}

We also need the following lemma.

\begin{lemma}\label{condensationsmallT}
For every $\mu>0$, there exists a temperature $T_1>0$, such that any minimizer of the grand canonical functional \eqref{def:grandcanfreeenergyfunctional} at temperatures $0\leq T\leq T_1$ and chemical potential $\mu$ has $\rho_0>0$.
\end{lemma}
\begin{proof}
Assume that a minimizer has $\rho_0=0$. This implies that its $\gamma$ satisfies
\begin{eqnarray}
\begin{aligned}
\cF^T(\gamma,0,0)\leq \inf_{\rho_0}\cF^T(0,0,\rho_0)=\frac{-\mu^2}{2\widehat{V}(0)}, \label{assumenocondensation}
\end{aligned}
\end{eqnarray}
since adding an $\alpha$ could only raise the energy (due to the monotonicity of the entropy \eqref{ineq:sga<sg0}).
We have 
\begin{align}
\cF^T(\gamma,0,0)&= \frac12 \int p^2\gamma(p)dp+\frac12\int\widehat{V}(p-q)\gamma(p)\gamma(q)dpdq \label{term1} \\ &+ \frac12 \widehat{V}(0)\rho_\gamma^2-\mu\rho_\gamma \label{term2} \\&+ \frac12 \int p^2\gamma(p)dp-TS(\gamma,0). \label{term3}
\end{align}
Clearly,
\begin{eqnarray}
\eqref{term2}\geq \frac{-\mu^2}{2\widehat{V}(0)}, \label{term2bound}
\end{eqnarray}
and \eqref{term3} can be bounded as in \eqref{noninteractingbound}, i.e.
\begin{eqnarray}
\eqref{term3}\geq -CT^{5/2} \label{term3bound},
\end{eqnarray}
where $C$ is a positive constant. Since $\widehat{V}(0)>0$ and $\widehat{V}\in C^1$, we can pick $p_0>0$ small enough such that $\min_{|p|\leq 2p_0}\widehat{V}(p)>0$, and 
\begin{eqnarray*}
\eqref{term1}\geq  \frac12 p_0^2\int_{|p|>p_0}\gamma(p)dp+\frac12\left(\min_{|p|\leq 2p_0}\widehat{V}(p)\right)\left(\int_{|p|\leq p_0}\gamma(p)dp\right)^2.
\end{eqnarray*}
The last expression can be minimized in $\int_{|p|\leq p_0}\gamma$, where we also take into account that it is less than $\rho_\gamma$. The lower bounds  we deduce are
\begin{eqnarray*}
&\frac12\min_{|p|\leq 2p_0}\widehat{V}(p)\rho_\gamma^2 \qquad &\text{if} \qquad \rho_\gamma \leq \frac{p_0^2}{2\min_{|p|\leq 2p_0}\widehat{V}(p)} \\
&\frac14 p_0^2\rho_\gamma \qquad &\text{if} \qquad  \rho_\gamma>\frac{p_0^2}{2\min_{|p|\leq 2p_0}\widehat{V}(p)}. 
\end{eqnarray*}
It follows that there exist $c_1,c_2>0$ depending only on $V$ such that
\begin{eqnarray}
\eqref{term1}\geq \min \{c_1 \rho_\gamma, c_2 \rho_\gamma^2\}. \label{term1bound}
\end{eqnarray}
Putting together \eqref{assumenocondensation}, \eqref{term2bound}, \eqref{term3bound} and \eqref{term1bound}, we see that any minimizer with $\rho_0=0$ has to satisfy
 \begin{eqnarray*}
  \min \{c_1 \rho_\gamma, c_2 \rho_\gamma^2\}\leq C T^{5/2}.
\end{eqnarray*}  
However, this means that there exist $c_3,c_4>0$ depending only on $V$ such that 
\begin{equation}
\label{conditioncritr}
\begin{aligned}
\frac{\partial \cE(\gamma,0,\rho_0)}{\partial \rho_0}\Big|_{\rho_0=0}=-\mu+\widehat{V}(0)\int\gamma+\int\widehat{V}\gamma&\leq -\mu+2\widehat{V}(0)\int\gamma\\
&\leq -\mu+\max\{c_3T^{5/2},c_4T^{5/4}\}.
\end{aligned}
\end{equation}
This implies the existence of a temperature $T_1$ depending on $\mu$ and $V$ such that this derivative is negative for all $0\leq T\leq T_1$, which means that there cannot be minimizers with $\rho_0=0$.
\qed
\end{proof}

\begin{proof}[Proof of Proposition \ref{limit0}]
We split the proof into several steps in which we obtain the different limits. For simplicity we use the notation $\int \gamma^T=:\rho^T_\gamma$ and $\rho^T:=\rho^T_\gamma+\rho_0^T$.\\

\textit{Step 1: Limit for $\rho^T_0$ and $\rho^T_\gamma$}. We will show that both these sequences are uniformly bounded.
Since we are dealing with minimizers, we have
\[
\frac{-\mu^2}{2\widehat{V}(0)}=\mathcal{F}^T\left(0,0,\frac{\mu}{\widehat{V}(0)}\right)\geq\mathcal{F}^0\left(\gamma^T,\alpha^T,\rho_0^T\right)-TS(\gamma^T,\alpha^T).
\]
Since by Proposition \ref{prop:T=0minseq} the entropy term converges to $0$ as $T\to0$, for $T$ small enough we have 
\begin{align*}
\frac{-\mu^2}{4\widehat{V}(0)}&\geq-\mu\rho^T+\frac{1}{2}\widehat{V}(0)\left(\rho^T\right)^{2} +\rho^T_{0}\int\widehat{V}(\gamma^T+\alpha^T)\\
&\geq \rho_0^T(-\mu-\frac{1}{2}\int\widehat{V})+\frac{1}{2}\widehat{V}(0)\left(\rho_0^T\right)^{2} -\mu\rho_\gamma^T+\frac{1}{2}\widehat{V}(0)\left(\rho_\gamma^T\right)^{2},
\end{align*}
where we have thrown out some positive terms and used the fact that $\gamma+\alpha\geq-\frac12$. This estimate implies that $\rho^T_0$ and $\rho^T_\gamma$ are uniformly bounded.
We can extract a limit by taking subsequences, so that from now on we have  $\rho^T_0\to\tilde{\rho}_0$ and $\rho^T_\gamma\to\tilde{\rho}_\gamma$. \\

\textit{Step 2: Limit for $\int\widehat{V}(\gamma^T+\alpha^T)$ and $\tilde{\rho}_0>0$}. It follows from Lemma \ref{condensationsmallT} that $\rho_0^T>0$ for $T$ small enough. This implies that the Euler--Lagrange equation in $\rho_0$ has to hold with equality for $T$ small enough:
\begin{equation}
-\mu+\int\widehat{V}(\gamma^T+\alpha^T)+\widehat{V}(0)\rho^T=0.
\label{rel1123}
\end{equation}
Since we know that $\rho^T$ has a limit as $T\to0$, the integral in the equation above will also have a limit. 

We now consider the following trial state: 
\begin{eqnarray*}
\begin{aligned}
\gamma &=\gamma_0\mathbbm{1}_{B_\epsilon} \\
\alpha &=-\sqrt{(\gamma_0(\gamma_0+1)}\mathbbm{1}_{B_\epsilon}, 
\end{aligned} 
\end{eqnarray*}
where $B_\epsilon$ denotes the ball with radius $\epsilon$ (which will be fixed later) centred at the origin.  We have
\bq
\begin{aligned}
\cF^0&\left(\gamma,\alpha,\mu(\widehat{V}(0))^{-1}-\gamma_0 |B_\epsilon|\right)-\cF^0\left(0,0,\mu(\widehat{V}(0))^{-1}\right)= \\ 
\\ &= \gamma_0\int\limits_{B_\epsilon} p^2 dp   +\left(\frac{\mu}{\widehat{V}(0)}-|B_\epsilon|\gamma_0\right)(\gamma_0-\sqrt{(\gamma_0(\gamma_0+1)})\int\limits_{B_\epsilon}\widehat{V}(p)dp \\
 &\quad+\frac{2\gamma_0^2+\gamma_0}{2}\iint_{B_\epsilon\times B_\epsilon}\widehat{V}(p-q)dpdq. \label{differenceT=0}
\end{aligned}
\eq
Assume that $\gamma_0$ is large enough, in particular $\gamma_0>1$. Then 
$$\gamma_0-\sqrt{(\gamma_0(\gamma_0+1)}=-\frac12+O(\gamma_0^{-1}).$$
We also choose the radius $\epsilon$ in such a way that 
$$|B_\epsilon|=\frac{\nu}{\gamma_0^2}$$
for a positive constant $\nu$. The fact that $\widehat{V}\in C^1$ and $\widehat{V}(0)>0$ imply that $\widehat{V}(p)\geq \frac12 \widehat{V}(0)$ on $B_\epsilon$ for $\gamma_0$ large enough. It follows that
\bq
\begin{aligned}
\eqref{differenceT=0}&\leq C\gamma_0 |B_\epsilon|^{5/3}-(\frac12-O(\gamma_0^{-1}))\left(\frac{\mu}{\widehat{V}(0)}-|B_\epsilon|\gamma_0\right)|B_\epsilon|\frac{\widehat{V}(0)}{2} \\ &\quad+\widehat{V}(0)|B_\epsilon|^2\frac{2\gamma_0^2+\gamma_0}{2}, \nn
\end{aligned}
\eq
where $C$ is a positive constant. Hence, for $\nu<\frac{\mu}{4\widehat{V}(0)}$ and $\gamma_0$ sufficiently large
\begin{align*}
\eqref{differenceT=0}&\leq  C\gamma_0^{-7/3}-\frac12\left(\frac{\widehat{V}(0)}{2}-O(\gamma_0^{-1})\right)\left(\frac{\mu}{\widehat{V}(0)}-\frac{\nu}{\gamma_0}\right)\frac{\nu}{\gamma_0^2}+\widehat{V}(0)\nu^2\frac{2\gamma_0^2+\gamma_0}{2\gamma_0^4} \notag\\
&= \left(\nu\widehat{V}(0)-\frac{\mu}{4}\right)\frac{\nu}{\gamma_0^2}+o(\gamma_0^{-2})<0. \label{differenceT=0<0}
\end{align*}
Also note that $\gamma_0>1$ implies $\mu(\widehat{V}(0))^{-1}-\gamma_0 |B_\epsilon|\geq 0$, which means that our choice of $\rho_0$ in \eqref{differenceT=0} was allowed.

Together with Proposition \ref{prop:T=0minseq}, this calculation implies that
\begin{equation}
\begin{aligned}
\label{rel1124}
\frac{-\mu^2}{2\widehat{V}(0)}&>F(0,\mu)=\lim_{T\rightarrow0}\mathcal{F}^0(\gamma^T,\alpha^T,\rho_0^T)\\
&\hspace{0.5cm}=\lim_{T\rightarrow0}\left[\int p^{2}\gamma^T+\frac{1}{2}\int\widehat{V}(p-q)\left[\alpha^T(p)\alpha^T(q)+\gamma^T(p)\gamma^T(q)\right]dpdq \right]\\
&\hspace{0.5cm}\quad+\left[-\mu\tilde{\rho}+\frac{1}{2}\widehat{V}(0)\tilde{\rho}^{2}\right] +\tilde{\rho}_{0}\lim_{T\rightarrow0}\int\widehat{V}(\gamma^T+\alpha^T).\\
\end{aligned}
\end{equation}
The first limit has to be non-negative and the term involving $\tilde{\rho}$ has to be bigger than or equal to $-\mu^2/2\widehat{V}(0)$. We therefore conclude that 
\begin{eqnarray}
\begin{aligned}
&\tilde{\rho}_0>0,\\
\lim_{T\rightarrow0}\int\widehat{V}(\gamma^T&+\alpha^T)=-C<0.
\label{obs00}
\end{aligned}
\end{eqnarray}\\
\textit{Step 3: Limits for $TA^T$ and $TB^T$}.
Recall from Section \ref{apgr} that the Euler--Lagrange equations of the functional lead to an expression for $\gamma^T$ in terms of the functions
\begin{eqnarray} 
\begin{aligned}
\label{defofATzero}
TA^T(p)&=p^2-\mu+\widehat{V}(0)\rho^T+\rho^T_0 \widehat{V}(p) + \widehat{V}\ast \gamma^T(p),\\
TB^T(p)&=\rho^T_0 \widehat{V}(p) + \widehat{V}\ast \alpha^T(p),\\
TG^T(p)&=\sqrt{\left(TA^T(p)\right)^2-\left(TB^T(p)\right)^2}.
\label{AandB}
\end{aligned}
\end{eqnarray}
We will establish a limit for these functions, and then prove that it leads to a limit for $\gamma^T$. Note that we only need to deal with the convolution terms since all other terms already have a limit or are constant in $T$. 

Our goal is a pointwise limit on the whole space, and a $C^2$-limit on the compact  $\{|p|\leq P_0\}$, where $P_0$ is given by Lemma \ref{largepboundgamma2}. 
Recall our assumption that $\widehat{V}$ is in $C^3(\mathbb{R}^3)$ and that all its derivatives up to third order are bounded. This implies that $\widehat{V}\ast \gamma^T$ and $\widehat{V}\ast \alpha^T$ are also in $C^3(\mathbb{R}^3)$ and, using the bounds \eqref{convValpha} on these quantities and the uniform bound on $\rho_\gamma^T$, that all derivatives up to third order are uniformly bounded in $T$. In particular, $\widehat{V}\ast \gamma^T$ and $\widehat{V}\ast \alpha^T$ are uniformly bounded with uniformly bounded derivatives, and the latter implies uniform equicontinuity. All this means that by a diagonal argument one can construct a pointwise limit on $\mathbb{R}^3$ (that is continuous) by selecting subsequences that converge on the rationals (see, e.g.\ Theorem I.26 in \cite{ReeSim1}). By the Arzel\`{a}--Ascoli theorem, this implies that taking further subsequences leads to a uniform limit on the compact $\{|p|\leq P_0\}$. We now repeat this last argument for the derivatives and second-order derivatives on $\{|p|\leq P_0\}$. We obtain uniform (continuous) limits for all derivatives up to second order. By uniform convergence these are indeed derivatives of the limit functions.

Summarizing, we have obtained limits $a$ and $b$ that are bounded and in $C^2(|p|\leq P_0)$ such that $TA^T\to a$ and  $TB^T\to b$ pointwise and also uniformly on $\{|p|\leq P_0\}$. We also note that by \eqref{rel1123}, \eqref{obs00} and \eqref{AandB}: $a\geq C>0$. By the Euler--Lagrange equations for $T>0$ we have $|TB^T|\leq TA^T$, so the limits also satisfy $|b|\leq a$. Hence $TG^T$ also has a pointwise limit $g=\sqrt{a^2-b^2}$ that is a bounded function.\\

\textit{Step 4: Limit for $\gamma^T$}.
As in Section \ref{apgr} we derive an expression for $\gamma^T$ in terms of $A^T$ and $G^T$ given by \eqref{defofATzero}. To make use of the limits we have obtained, we write it as follows:
\begin{equation}
\gamma^T=\frac{TA^T-TG^T+e^{-\frac{1}{T}(TG^T)}\left(TA^T+TG^T\right)}{2TG^T\left(1-e^{-\frac{1}{T}(TG^T)}\right)}.
\label{gammaT}
\end{equation}
We conclude that pointwise
\begin{equation}
\gamma^T\xrightarrow{T\to 0} \frac{a-g}{2g}=:\tilde{\gamma},
\label{convgam}
\end{equation}
which is easy to see when $g>0$, but since $a>0$ it is also true for $g=0$ (with the understanding that $\gamma^T=+\infty$ at such points). We would nonetheless like to prove that $g=0$ actually cannot happen. 

First note that $g(p)$ is bounded away from $0$ for $|p|\geq P_0$ by the bound in Lemma \ref{largepboundgamma2} and the fact that $a\geq C>0$. Now suppose that $g^2(p_0)=(a^2-b^2)(p_0)=0$ for some $|p_0|< P_0$. We know that $a^2$ and $b^2$ are $C^2$ around $p_0$ and that $b^2\leq a^2$. Therefore, $a^2-b^2$ has to behave like $(p-p_0)^{2}+o(p-p_0)^{2}$ around $p_0$. Since $a\geq C>0$, we see that $\tilde{\gamma}$ has to go to infinity like $|p-p_0|^{-1}$ or faster. If we assume that $p_0\neq0$, this implies that $\tilde{\gamma}$ is non-integrable, which, by Fatou's lemma, contradicts the pointwise convergence:
\[
\int\tilde{\gamma}\leq\liminf_{T\to0}\int\gamma^T=\tilde{\rho}_\gamma<\infty.
\]
We therefore conclude that $g(p)$ cannot be zero for $p\neq0$. However, using \eqref{rel1123} and \eqref{AandB} we can calculate that 
\[
g(0)=\sqrt{(a(0)-b(0))(a(0)+b(0))}=\sqrt{-4\tilde{\rho}_0\widehat{V}(0)\left[\lim_{T\to0}\int\widehat{V}\alpha^T\right]}>0,
\]
where the inequality holds by \eqref{obs00}.
We can now conclude that $g\neq0$. Since it is continuous, it has to be bounded away from zero on the compact $|p|\leq P_0$, and combined with our previous observation, everywhere. 

We now analyse the expression \eqref{gammaT} and conclude that the convergence \eqref{convgam} is actually uniform on $\{|p|\leq P_0\}$. For this we use the following facts: a sum preserves uniform convergence; a product preserves uniform convergence given that the limit functions are bounded; a composition $g\circ f_n$ preserves uniform convergence (of the $f_n$) if $g$ is uniformly continuous in the region where $f_n$ takes values. Since it is necessary to apply this last fact to the function $x\mapsto 1/x$, it is crucial that $g$ is bounded away from 0. 

We can finally prove that $\gamma^T\to\tilde{\gamma}$ in $L^1$. The uniform convergence implies $L^1$-convergence on $\{|p|\leq P_0\}$. By Lemma \ref{largepboundgamma2}, we have also uniform boundedness by an $L^1$-function on $\{|p|> P_0\}$. Applying the Dominated Convergence Theorem to that region, we conclude that $\gamma^T\to\tilde{\gamma}$ in $L^1$. The pointwise convergence obtained before also implies $\tilde\gamma\geq0$, and by Fatou's lemma, $\int p^2\tilde{\gamma}dp<\infty$.\\

\textit{Step 5: Limit for $\alpha^T$}.
As before, we use relations that are known to hold for $T>0$ to conclude convergence:
\begin{align*}
\beta^T&=\sqrt{\left(\gamma^T+\frac{1}{2}\right)^2-\left(\alpha^T\right)^2}=\frac{1+e^{-\frac{1}{T}(TG^T)}}{2(1-e^{-\frac{1}{T}(TG^T)})}\xrightarrow{T\rightarrow 0} \frac{1}{2}\label{beta12} \\
\alpha^T&=-\beta^T\frac{TB^T}{TG^T}\xrightarrow{T\rightarrow 0}-\frac{b}{2g}=:\tilde{\alpha}.
\end{align*}
Again, the convergence holds pointwise everywhere and uniformly on $\{|p|\leq P_0\}$. The uniform convergence implies $L^2$-convergence on $\{|p|\leq P_0\}$. Since $\left(\alpha^T\right)^2\leq\gamma^T(\gamma^T+1)$, Lemma \ref{largepboundgamma2} leads to an uniform $L^2$-bound on the $\alpha^T$ for $\{|p|\geq P_0\}$. Hence, $L^2$-convergence also holds in this region by the Dominated Convergence Theorem. Also note that $\beta^T\to1/2$ implies that $\tilde{\alpha}^2=\tilde{\gamma}(\tilde{\gamma}+1)$. We have now proved all the claims in the proposition.
\qed
\end{proof}

It remains to prove Corollary \ref{cor:structureminimizerT=0}.
\begin{proof}[Corollary \ref{cor:structureminimizerT=0} for the grand canonical functional.]
Our goal will be to show that any minimizer at $T=0$ has to satisfy $\alpha^2=\gamma(\gamma+1)$ using elements from the proof above.
The corollary then follows from Theorem 10.4 in \cite{Solovej-notes}, which states that the 1-pdm $\Gamma$ corresponds to pure quasi-free states if and only if
$$\Gamma \mathcal{S} \Gamma=-\Gamma$$
(cf. \eqref{1pdm} and \eqref{sj} for definitions). This is indeed satisfied if $\alpha^2=\gamma(\gamma+1)$.

Note that $\mu\leq0$ is easy, since the minimizer is $(\gamma,\alpha,\rho_0)=(0,0,0)$ as explained at the start of this section. 
For $\mu>0$, we can consider \eqref{rel1124} directly at $T=0$ (i.e.\ without the limits) to conclude that any minimizer has
\[
\rho_0>0, \hspace{1.5cm}\int\widehat{V}(\gamma+\alpha)<0.
\]
This implies that \eqref{rel1123} holds, and so minimizers have
\[
\frac{\partial\mathcal{F}}{\partial\gamma}=p^2-\int\widehat{V}(\gamma+\alpha) +\rho_0\widehat{V}(p)+ \widehat{V}\ast \gamma(p)>0.
\]
If $\alpha^2<\gamma(\gamma+1)$, this derivative equals 0 and we obtain a contradiction.
\qed
\end{proof}

\subsection{The canonical case} 
\label{thecancaselabel}
We would now like to prove the existence of $T=0$  minimizers for the canonical problem. Recall that for fixed $\rho\geq0$ and $T\geq0$ the functional reads  
\bq
\begin{aligned}
\cF^{\rm{can}}(\gamma,\alpha,\rho-\rho_\gamma)&= \int p^{2}\gamma(p)dp-TS(\gamma,\alpha)+\frac12\widehat{V}(0)\rho^{2}\\ &+\left(\rho-\int\gamma\right)\int\widehat{V}(p)\left(\gamma(p)+\alpha(p)\right)dp \\
&+\frac{1}{2}\iint\widehat{V}(p-q)\left(\alpha(p)\alpha(q)+\gamma(p)\gamma(q)\right)dpdq. 
\end{aligned} \nn
\eq
\begin{proof}[Proof of Theorem \ref{thm:existencecanonicalzeroT}.] We follow the same strategy as in the grand canonical case. The same argument as in Proposition \ref{prop:T=0minseq} implies that canonical, positive temperature minimizers at fixed $\rho$ form a minimizing sequence for the $T=0$ problem with that $\rho$. 

We have
\begin{equation}
\label{ELeqnscan}
 \begin{aligned}
\frac{\partial\mathcal{F}^{\rm{can}}}{\partial\gamma}&=p^2+(\rho-\rho_\gamma) \widehat{V}(p)-\int\widehat{V}(\gamma+\alpha) + \widehat{V}\ast \gamma(p)-T\frac{\gamma+\frac{1}{2}}{\beta}\ln{\frac{\beta+\frac{1}{2}}{\beta-\frac{1}{2}}}\\
\frac{\partial\mathcal{F}^{\rm{can}}}{\partial\alpha}&=(\rho-\rho_\gamma)\widehat{V}(p) + \widehat{V}\ast \alpha(p)+T\frac{\alpha}{\beta}\ln{\frac{\beta+\frac{1}{2}}{\beta-\frac{1}{2}}}.
\end{aligned}
\end{equation}
To see that these expressions are equal to zero for minimizers, we repeat the argument in Lemmas \ref{lem:gammabigger0} and \ref{lem:apriorialpha}, but one extra ingredient is needed since $\rho_\gamma\leq\rho$ provides an extra constraint compared to the grand canonical case. We therefore apply Theorem \ref{thm:phasetrancan} (proved in the next section), which states that minimizers will have $\rho_0>0$ for sufficiently low temperatures. As a consequence, we arrive at the same bound as in Lemma \ref{largepboundgamma2}.

We now repeat the proof of Proposition \ref{limit0}. Step 1 simplifies since $\rho^T_\gamma\leq\rho$ provides the required bound. For step 2, we first note that there is no equivalent to \eqref{rel1123} in this case, but we can take a further subsequence to ensure that $\int\widehat{V}(\gamma^T+\alpha^T)$ has a limit. We then repeat the trial state argument (with $\mu/\widehat{V}(0)$ replaced with $\rho$), and it leads to the same conclusion as in the grand canonical case, that is
\begin{eqnarray*}
\begin{aligned}
&\tilde{\rho}_0:=\lim_{T\to 0}(\rho-\rho_\gamma^T)>0,\\
\lim_{T\rightarrow0}&\int\widehat{V}(\gamma^T+\alpha^T)=-C<0.
\end{aligned}
\end{eqnarray*}
The canonical $TA^T$ reads   
\bq
TA^T=p^2+\rho_0^T\widehat{V}(p)-\int\widehat{V}(\gamma^T+\alpha^T) + \widehat{V}\ast \gamma^T(p), \nn
\eq
which is really the same as \eqref{rel1123} combined with \eqref{defofATzero}.
We then repeat the remaining steps in the proof of Proposition \ref{limit0} to reach similar conclusions. To finish, we proceed as in the proof of Theorem \ref{thm:existencezeroT}. The conclusion of Corollary \ref{cor:structureminimizerT=0} for the canonical functional follows in an identical way.
\qed
\end{proof}

\section{Phase transition and the grand canonical phase diagram}\label{sec:exi_phase_tran}
We start by proving Theorem \ref{thm:BECvsSF}, which states that there is only one kind of phase transition in the system.
This holds for both the canonical and the grand canonical functional.

\begin{proof}[Proof of Theorem \ref{thm:BECvsSF}.] \textit{Step 1}. Let $T>0$. Since 
$$\int\widehat{V}(p-q)\alpha(p)\alpha(q)dpdq=\int V(x)|\check\alpha(x)|^2dx\geq 0,$$
and $S(\gamma,\alpha)< S(\gamma,0)$ for $\alpha\not\equiv0$, we directly see from the definition of the functionals \eqref{def:grandcanfreeenergyfunctional} and \eqref{def:canonicalfreeenergyfunctional}  that $\rho_0=0$ implies $\alpha\equiv0$.

Let $T>0$. Recall from the proof of the existence of minimizers that the Euler--Lagrange equation for $\alpha$ is satisfied:
$$\int\widehat{V}(p-q)\alpha(q)dq+\rho_{0}\widehat{V}(p)+T\frac{\alpha(p)}{\beta(p)}\ln \frac{\beta(p)+\frac12}{\beta(p)-\frac12}=0$$
for both functionals.
Thus $\alpha\equiv0$ implies $\rho_0=0$ as long as $\widehat{V}(p)>0$ on some set of positive measure, which is the case since $\widehat{V}(0)>0$ and $\widehat{V}\in C^1$.\\
 
\textit{Step 2.} Let $T=0$. For $\mu<0$ (grand canonically) or $\rho=0$ (canonically), we know that the minimizers have $\rho_0=\alpha=0$, so there is nothing to prove.

For $\mu>0$ or $\rho>0$, we know that $\rho_0>0$ by Theorems \ref{thm:phasetrangrandcan} and \ref{thm:phasetrancan} respectively. Grand canonically, we have shown in Corollary \ref{cor:structureminimizerT=0} that $\alpha\not\equiv0$, which followed from the trial state argument in step 2 of the proof of Proposition \ref{limit0}. As pointed out in the proof of Theorem \ref{thm:existencecanonicalzeroT}, a similar argument holds for the canonical case, and we again find $\alpha\not\equiv0$.
\qed
\end{proof}

We now prove that there indeed exists a phase transition in the model.

\begin{proof}[Proof of Theorem \ref{thm:phasetrangrandcan}.]
Note that the second part of the statement is proved in Lemma \ref{condensationsmallT}. It remains to show that there is no condensation for high temperatures. 

The proof is based on two inequalities: an upper and a lower bound. The upper bound shows that for sufficiently large $T$ there exists a positive constant $C$ depending on $\mu$ and $V$ such that
\begin{align}
\inf_{\gamma}\cF(\gamma,0,0)\leq -C T^2 \ln T + O(T^2). \label{upperboundinfgamma00}
\end{align}
The lower bound shows that any minimizer $(\gamma,\alpha,\rho_0)$ with $\rho_0>0$ has to satisfy
\[
\cF(\gamma,\alpha,\rho_0)\geq -\tilde{C} T\ln T +O(T) 
\]
for sufficiently large $T$ and $\tilde{C}$ depending on $\mu$ and $V$. Hence, the minimizer has $\rho_0=0$ and $\alpha\equiv 0$ for $T$ large enough.\\

\textit{Upper bound.} We start by proving \eqref{upperboundinfgamma00}. Note that
\bq
\begin{aligned}
\cF(\gamma,0,0)& \leq \int(p^2-\mu)\gamma(p)dp+ \widehat{V}(0)\rho_\gamma^2\\&+ T\int\left[\gamma(p)\ln \gamma(p)-(\gamma(p)+1)\ln(\gamma(p)+1)\right]dp. \nn
\end{aligned}
\eq
To obtain an upper bound, we evaluate the right-hand side of the inequality above using the trial state
\[
\gamma_{\delta}(p)=\left(e^{\frac{p^2+\delta}{T}}-1\right)^{-1}, 
\]
where $\delta$ is a positive constant, so that 
\begin{align}
\cF(\gamma_{\delta},0,0)&\leq T\int\ln\left(1-e^{\frac{-(p^2+\delta)}{T}}\right)dp-(\mu+\delta)\int \left(e^{\frac{p^2+\delta}{T}}-1\right)^{-1}dp \notag \\ & +\widehat{V}(0)\left(\int \left(e^{\frac{p^2+\delta}{T}}-1\right)^{-1}dp\right)^2. \label{upperboundgamma00}
\end{align}
Note that
\begin{align}
T\int\ln\left(1-e^{\frac{-(p^2+\delta)}{T}}\right)dp\leq -Te^{-\delta/T}\int e^{-p^2/T}dp=-C_0 T^{5/2}e^{-\delta/T}, \label{upperboundfreefreeenergy}
\end{align}
where $C_0=(2\pi)^{-2}\int_{0}^{\infty}\sqrt{s}e^{-s}ds<\infty$ (recall our convention for the measures $dp$ and $ds$ explained above \eqref{def:auxminimization}). Also
\begin{align*}
\int\gamma_{\delta}(p)dp=T^{3/2}e^{-\delta/T}\int_{0}^{\infty}\frac{(2\pi)^{-2}\sqrt{s}}{e^s-e^{-\delta/T}}ds. 
\end{align*}
Clearly,
\begin{eqnarray}
C_0 \leq\int_{0}^{\infty}\frac{(2\pi)^{-2}\sqrt{s}}{e^s-e^{-\delta/T}}ds\leq  \int_{0}^{\infty}\frac{(2\pi)^{-2}\sqrt{s}}{e^s-1}ds=: C_1, \label{boundsintgammadelta}
\end{eqnarray}
and so
\begin{eqnarray}
\label{someref00234}
  C_0 T^{3/2}e^{-\delta/T}  \leq\int\gamma_{\delta}(p)\leq C_1 T^{3/2}e^{-\delta/T}.
\end{eqnarray}
Using \eqref{upperboundfreefreeenergy} and \eqref{boundsintgammadelta} in \eqref{upperboundgamma00} we obtain 
\begin{align*}
\cF(\gamma_{\delta},0,0)\leq -C_0 T^{5/2}e^{-\delta/T}-C_0 (\delta+\mu)T^{3/2}e^{-\delta/T}+C_1^2\widehat{V}(0) T^3e^{-2\delta/T}. 
\end{align*}
We now choose $\delta=\frac12 T\ln(T)$. Then $e^{-\delta/T}=T^{-1/2}$, which implies
\[
\inf_{\gamma}\cF(\gamma,0,0)\leq  -\frac{C_0}{2}T^2 \ln T+(C_1^2\widehat{V}(0)-C_0)T^{2}-C_0 \mu T, 
\]
and we arrive at the desired upper bound \eqref{upperboundinfgamma00}.\\

\textit{Lower bound.} 
Any minimizer $(\gamma,\alpha,\rho_0)$ has to satisfy
\begin{align*}
\cF(\gamma,0,0)\geq\cF(\gamma, \alpha, \rho_0),
\end{align*}
which, using monotonicity of the entropy in $\alpha^2$, the fact that $\gamma+\alpha\geq-1/2$, and our assumption $\rho_0>0$, implies that 
\begin{align}
\int\gamma(p)dp\leq \frac{\mu+\frac12\int\widehat{V}(p)dp}{\widehat{V}(0)}:=A>0,\label{gammaintbound}
\end{align}
where the constant $A$ is positive and only depends on $\mu$ and $V$. Combining this knowledge with the aforementioned facts in the same way, we obtain
\begin{align*}
\cF(\gamma,\alpha,\rho_0)& \geq\int p^{2}\gamma(p)dp-TS(\gamma,0)-\mu A-\rho_0A\widehat{V}(0)+\frac12 \rho_0^2 \widehat{V}(0).
\end{align*}
A lower bound for the terms involving $\rho_0$ can be calculated explicitly. Using \eqref{gammaintbound} again, we obtain for any $\delta\geq0$:
\begin{align*}
\cF(\gamma,\alpha,\rho_0)& \geq\int (p^{2}+\delta)\gamma(p)dp -TS(\gamma,0)-\delta A-\mu A-\frac12A^2\widehat{V}(0).
\end{align*}
To obtain a lower bound, we now minimize the expression involving $\gamma$, which leads to the bound
\[
\cF(\gamma,\alpha,\rho_0) \geq T\int\ln\left(1-e^{\frac{-(p^2+\delta)}{T}}\right)dp  -\delta A-\mu A-\frac12A^2\widehat{V}(0).
\]
Since 
\begin{align*}
\ln\left(1-e^{\frac{-(p^2+\delta)}{T}}\right) \geq -\frac{1}{e^{\frac{p^2+\delta}{T}}-1},
\end{align*}
one has 
\begin{align*}
T\int\ln\left(1-e^{\frac{-(p^2+\delta)}{T}}\right)dp\geq -T\int\frac{dp}{e^{\frac{p^2+\delta}{T}}-1}\geq-C_1 T^{5/2}e^{-\delta/T},
\end{align*}
where we use \eqref{someref00234}. Thus, choosing $\delta=\frac32 T\ln T$ we arrive at 
\begin{align*}
\cF(\gamma,\alpha,\rho_0) \geq -C_1 T-\frac32 A T\ln T-\mu A-\frac12A^2\widehat{V}(0),
\end{align*}
which completes the proof of the lower bound.
\qed
\end{proof}

We now prove the existence of a phase transition for the canonical problem.

\begin{proof}[Proof of Theorem \ref{thm:phasetrancan}.] \textit{Step 1.}
Let $\rho>0$ be fixed and let $C_0$ be a constant depending on $V$ and $\rho$ that will be fixed later on.
Consider
\bq
U=\left\{|p|>\sqrt{2\rho\widehat{V}(0)+C_0(\rho,\widehat{V})}\right\}\subset\mathbb{R}^3.\nn
\eq
There exists a temperature $T_4$ depending only on $\rho$ and $V$ such that for $T>T_4$, we have 
\bq
\rho<\int_U \frac1{e^{2p^2/T}-1}dp.  \label{eq:rhoboundlargeT}
\eq
We will prove that \eqref{eq:rhoboundlargeT} implies that $\rho_0=0$ for the minimizer.

To prove this claim, consider any $(\gamma,\alpha,\rho-\int\gamma)$ with $\int\gamma\leq\rho$. Note that by \eqref{eq:rhoboundlargeT} there exists a subset $V\subset U$ with positive measure such that
\bq
\gamma(p)\big|_V<\frac1{e^{2p^2/T}-1}. \label{eq:gammauperboundV}
\eq
Recall the functional derivative of the canonical functional in \eqref{ELeqnscan}. Using the fact that the gamma-derivative of the entropy is monotone increasing in $\alpha^2$ in the first step and \eqref{finaleqref?} in the second (which defines $C_0$), we obtain 
\bq
\begin{aligned}
\frac{\partial \cF^{\rm{can}}}{\partial \gamma} &\leq p^2-T\ln\left(\frac{\gamma(p)+1}{\gamma(p)}\right)+\left(\rho-\rho_\gamma\right)\widehat{V}(p)+\widehat{V}*\gamma(p)-\int\widehat{V}(\gamma+\alpha) \\
&\leq p^2-T\ln\left(\frac{\gamma(p)+1}{\gamma(p)}\right)+2\rho\widehat{V}(0)+C_0. \nn
\end{aligned}
\eq
The bound \eqref{eq:gammauperboundV} implies that on $V\subset U$ we have
\bq
\frac{\partial \cF^{\rm{can}}}{\partial \gamma}\big|_{(\gamma,\alpha)}<0. \nn
\eq
In particular, we can lower the energy corresponding to any $\gamma$ with $\int\gamma\leq\rho$ by increasing it on some set of non-zero measure. However, this can only be done up to the point where $\int\gamma=\rho$. We therefore conclude that the minimizer will have to satisfy this, and hence $\rho_0=0$, which proves the claim.\\

\textit{Step 2.} We will now show that all $(\gamma,0,0)$ with $\int\gamma=\rho>0$ have a higher energy than $(0,0,\rho)$ for $0\leq T<T_3$, where $T_3>0$ is a constant temperature depending on $\rho$ and $V$. Since adding an $\alpha$ can never decrease the energy when $\rho_0=0$, this suffices.
We have
\begin{align*}
\mathcal{F}^{\rm{can}}(\gamma,0,0)=& \frac12 \int p^2\gamma(p)dp-TS(\gamma,0)+\frac12\widehat{V}(0)\rho^2\\
&+ \frac12 \int p^2\gamma(p)dp+\frac12\int\widehat{V}(p-q)\gamma(p)\gamma(q)dpdq\\
&\geq -CT^{5/2}+\frac12\widehat{V}(0)\rho^2+\min\left\{c_1\rho,c_2\rho^2\right\},
\end{align*}
where in the last step we used an argument similar to the one given in \eqref{term1bound}. Note that the last term is strictly positive and that it only depends on $\rho$ and $V$. 

This can be combined with 
\bq
\begin{aligned}
\mathcal{F}^{\rm{can}}(0,0,\rho)= \frac12\widehat{V}(0)\rho^2
\end{aligned}\nn
\eq
to give the estimate
\bq
\begin{aligned}
\mathcal{F}^{{\rm{can}}}(\gamma,0,0)-\mathcal{F}^{\rm{can}}(0,0,\rho)\geq 
\min\left\{c_1\rho,c_2\rho^2\right\}-CT^{5/2}.
\end{aligned}\nn
\eq
Since the first term is positive and only depends on $\rho$ and $V$, we see that this implies the existence of a $T_3>0$ as described above.
\qed
\end{proof}

What remains to be done is to determine the grand canonical phase diagram from Figure \ref{fig:phasediagram}. Most of the work has already been done. We will now collect some results and see how this diagram has been obtained.

For $\mu>0$, we have Theorem \ref{thm:phasetrangrandcan} and Lemma \ref{condensationsmallT}. Note that \eqref{conditioncritr} determines the lower bound of the region with the lighter shade of blue.  The bounds derived in the proof of Theorem \ref{thm:phasetrangrandcan} determine an upper bound on this region, but it does not go to 0 when $\mu$ does. To get the behaviour shown in Figure \ref{fig:phasediagram}, we need Theorem \ref{thm:grandcancrittemp}.

The case $T=0$ and $\mu\leq 0$ has been explained at the beginning of Subsection \ref{subsect:gcT=0}. By an argument similar to Lemma \ref{lem:gammabigger0}, we know that $\gamma>0$ for $T>0$. What remains to be shown is that there is no condensation for $T>0$ and $\mu\leq 0$. This follows from the fact that $\rho_0>0$ would imply
\bq
\begin{aligned}
\cF(\gamma,\alpha,\rho_0)&> \cF(\gamma,0,0) +\rho_0\int\widehat{V}(p) \alpha(p)dp+\frac12\int  \widehat{V}(p-q)\alpha(q) \alpha(p)dpdq \\ &\quad + \frac12\widehat{V}(0)\rho_0^2 + \widehat{V}(0)\rho_\gamma\rho_0 \\ & > \cF(\gamma,0,0) +\frac12\int \widehat{V}(p-q) (\alpha+\rho_0\delta)(p)(\alpha+\rho_0\delta)(q)dpdq \\ & \geq \cF(\gamma,0,0), \nn
\end{aligned}
\eq
where $\delta$ denotes the Dirac delta distribution. Hence $\rho_0=0$ for $\mu\leq0$. The conclusions for $\alpha$ follow from Theorem \ref{thm:BECvsSF}.

\appendix

\section{Derivation of the functional}\label{app:derivation}
\subsection{Bogoliubov trial states}
Let $\cH$ be a complex, separable Hilbert space with inner product $\langle\cdot,\cdot\rangle$, which is linear in the second variable and anti-linear in the first, and let $\Gamma_{\rm{s}}(\cH)$ be the bosonic Fock space related to $\cH$. 

Let $\mathcal{O}$ be the algebra of physical observables represented by bounded operators on $\Gamma_{\rm{s}}(\cH)$. A \textit{state} $\omega: \mathcal{O}\rightarrow \mathbb{C}$ of a quantum system is then identified with a positive semi-definite trace class operator $G$ on $\Gamma_{\rm{s}}(\cH)$  with $\text{Tr}(G)=1$ in the following way:
\begin{equation}
\omega(O)=\text{Tr}(OG) \quad \text{for all bounded} \quad O\in \mathcal{O}. \label{state}
\end{equation}
The operator $G$ is sometimes called the \textit{density matrix}.
The dual space $\mathcal{H}^{*}$ can be identified with $\mathcal{H}$ by the anti-unitary operator $J:\mathcal{H}\rightarrow \mathcal{H}^{*}$ defined by 
$$J(f)(g)=\langle f,g \rangle_{\mathcal{H}}, \quad \text{for all} \quad f,g\in\mathcal{H}.$$
If $a^{*}(f)$ and $a(g)$ are the usual bosonic creation and
annihilation operators on $\Gamma_{\rm{s}}(\cH)$ satisfying the
canonical commutation relations (CCR)
$$[a(g),a^{*}(f)]=(g,f), \quad [a^*(g),a^*(f)]=0, \quad [a(g),a(f)]=0 \quad \forall f,g\in\mathcal{H},$$
then one can introduce the \textit{field} or \textit{generalized
  creation} and \textit{annihilation} operators on
$\mathcal{H}\oplus\mathcal{H}^{*}$ by
\begin{align*} 
A(f\oplus Jg) &= a(f)+a^{*}(g),\\
A^{*}(f\oplus Jg) &= a^{*}(f)+a(g), \quad  \forall f,g\in\mathcal{H}. 
\end{align*}
By defining
\begin{equation}
\mathcal{S} =
 \begin{pmatrix}
  1 & 0 \\
  0 & -1
 \end{pmatrix} \quad \text{and} \quad
 \mathcal{J} =
 \begin{pmatrix}
  0 & J^{*} \\
  J & 0
 \end{pmatrix} \label{sj}
\end{equation}
one can express the CCR and conjugate relations in the following way:
$$A^{*}(F_{1})=A(\mathcal{J}F_{1}), \quad [A(F_{1}), A^{*}(F_{2})]=\langle F_{1},\mathcal{S}F_{2}\rangle \quad \text{for all} \quad F_{1},F_{2}\in\mathcal{H}\oplus\mathcal{H}^{*}.$$
We can now define the \textit{(generalized) one-particle density matrix} (\textit{1-pdm}) $\Gamma:\mathcal{H}\oplus\mathcal{H}^{*}\rightarrow \mathcal{H}\oplus\mathcal{H}^{*}$ of a state $\omega$ by 
$$\langle F_{1},\Gamma F_{2} \rangle=\omega(A^*(F_{2})A(F_{1})) \quad \text{for all} \quad F_{1},F_{2}\in\mathcal{H}\oplus\mathcal{H}^{*}.$$
Thus a 1-pdm can be written as
\begin{eqnarray}
\Gamma =
 \begin{pmatrix}
  \gamma & \alpha \\
  J\alpha J & 1+J\gamma J^*
 \end{pmatrix}, \label{1pdm}
\end{eqnarray}
where $\gamma:\mathcal{H}\rightarrow \mathcal{H}$ and $\alpha:\mathcal{H}^*\rightarrow \mathcal{H}$ are linear operators defined by
$$\langle f,\gamma g\rangle=\omega(a^*(g)a(f)), \quad \langle f,\alpha Jg\rangle=\omega(a(g)a(f)) \quad \forall f,g\in\mathcal{H}.$$
The definitions above imply in particular that states with finite particle number expectation are those for which $\gamma$ is trace class.\\
We shall now recall the notion of \textit{quasi-free states}. For our purpose a quasi-free state $\omega$ will be a state satisfying Wick's Theorem. In particular
\bq
\begin{aligned}
\omega(a_1^{\#} a_2^{\#} a_3^{\#} a_4^{\#})&=\omega(a_1^{\#} a_2^{\#})\omega(a_3^{\#} a_4^{\#})+\omega(a_1^{\#} a_4^{\#})\omega(a_2^{\#} a_3^{\#})+\omega(a_1^{\#} a_3^{\#})\omega(a_2^{\#} a_4^{\#}), \nn
\end{aligned}
\eq
where $a^{\#}$ is either $a$ or $a^*$. Furthermore, for any $m$ we have
\begin{align*}
\omega(a_1^{\#}\ldots a_{2m+1}^{\#})=0.
\end{align*}
If one considers a Bose system, one should extend the class of variational states by including so-called \textit{coherent states}. These states are used to describe the condensate fraction (for an explanation see e.g. \cite{Solovej-06}).

The mathematical implementation of that idea relies on the fact that for every $\phi\in\mathcal{H}$ there exists a unitary operator $\mathbb{U}_{\phi}:\Gamma_{\rm{s}}(\cH)\rightarrow\Gamma_{\rm{s}}(\cH)$ such that
$$\mathbb{U}^{*}_{\phi}a(f)\mathbb{U}_{\phi}=a(f)+\langle f,\phi \rangle \quad \forall f\in\mathcal{H}.$$ 
We may now describe the Bogoliubov variational states. Let $\omega_{\gamma,\alpha}$ be the quasi-free state with the 1-pdm $\Gamma_{\gamma,\alpha}$ and let $\phi\in\mathcal{H}$. The \textit{Bogoliubov variational state} $\omega_{\gamma,\alpha,\phi}$ is defined by 
\begin{equation}
\omega_{\gamma,\alpha,\phi}(O):=\omega_{\gamma,\alpha}(\mathbb{U}^{*}_{\phi}O\mathbb{U}_{\phi}) \quad \text{for all} \quad O\in\mathcal{O}. \label{bogvarstate}  
\end{equation}

\subsection{The Hamiltonian part}
Having introduced Bogoliubov variational states we will now turn to the derivation of the functional. Our model is based on the grand canonical Hamiltonian of the form
\begin{equation}
H=T+U=\sum_p (p^2-\mu) a_p^{*} a_p+\frac{1}{2L^3}
\sum_{p,q,k}\widehat{V}(k)a_{p+k}^{*}a_{q-k}^{*}a_q a_p, \label{secondquantizedhamiltonian}
\end{equation}
where the summation is taken over momenta $p,k,q \in \frac{2\pi}{L}\mathbb{Z}^{3}$. Here $a_p=a(L^{-3/2}e^{ipx})$.

Note that \eqref{secondquantizedhamiltonian} is the second quantization (in the plane wave basis) of the translation invariant grand canonical $N$-body Hamiltonian 
\[
H_N=\sum_{i=1}^{N}-\Delta^{L}_i+\sum_{i<j}V^{L}(x_i-x_j)  
\]
defined on $L^{2}_{\text{sym}}(\Lambda^{N})$, where $\Lambda=[-\frac{L}{2},\frac{L}{2}]^{3}$ is the physical space on which we impose periodic boundary conditions. The Laplacian is supposed to have periodic boundary conditions on $\Lambda$. The function $V^{L}$ is the periodized potential given by 
$$V^{L}(x)=\sum_{n\in \mathbb{Z}^{3}}V(x+nL).$$
We also have
$$\mathbb{U}^{*}_{\phi}a_{p}\mathbb{U}_{\phi}=a_{p}+\langle L^{-3/2}e^{ipx},\phi(x)\rangle.$$
Bogoliubov's $c$-number substitution (\cite{LieSeiYng-05}) is then implemented by choosing $\phi(x)$ to be a constant function equal to $\sqrt{\rho_0}$, where, as mentioned in the introduction, $\rho_0\geq 0$ has the interpretation of being the condensate density. Thus 
\[
\mathbb{U}^{*}_{\phi}a_{p}\mathbb{U}_{\phi}=a_{p}+\delta_{p,0}\sqrt{\rho_0}\sqrt{|\Lambda|}.
\]
According to \eqref{1pdm} we define 
\begin{eqnarray*}
\gamma(k):=\omega_{\gamma,\alpha}(a^*_k a_k), \qquad \text{and} \qquad \alpha(k):=\omega_{\gamma,\alpha}(a_k a_{-k}).
\end{eqnarray*}
Without loss of generality, we assume that our trial states satisfy 
$$\omega_{\gamma,\alpha}(a_k a_{-k})=\omega_{\gamma,\alpha}(a^*_k a^*_{-k}),$$
so that $\alpha$ takes real values.
A convexity argument shows that minimizing quasi-free states have this property and so it can indeed be assumed without loss of generality. We also restrict to translation-invariant quasi-free states, which does limit the model and implies that general results about minimization problems like \cite{quasifree} no longer hold.
A straightforward calculation, using the properties of quasi-free states and translation invariance of the system, then implies that 

\begin{eqnarray*}
\begin{aligned}
\omega_{\gamma,\alpha,\sqrt{\rho_0}}(H)&=\sum_{p}(p^2-\mu)\gamma(p)-\mu |\Lambda|\rho_0+\frac{\widehat{V}(0)}{2|\Lambda|}\sum_{p,q}\gamma(p)\gamma(q)\\ &+ \frac{1}{2|\Lambda|}\sum_{p,q}\left[\widehat{V}(p-q)\left(\alpha(p)\alpha(q)+\gamma(p)\gamma(q)\right)\right]\\
&+\frac{\rho_0^2|\Lambda|\widehat{V}(0)}{2}+\widehat{V}(0)\rho_0 \sum_{k}\gamma(k)+\rho_0\sum_{k}\widehat{V}(k)\left(\gamma(k)+\alpha(k)\right).
\end{aligned}
\end{eqnarray*}
The \textit{thermodynamic free energy} (per volume), $\mathcal{F}$, 
of a state $\omega$ at temperature $T\geq0$ and chemical potential $\mu\in \mathbb{R}$ is defined as 
\begin{equation*}
\mathcal{F}(\omega)=\frac{1}{|\Lambda|}\Big(\omega(H)-TS(\omega)\Big).
\end{equation*}
Taking the informal macroscopic limit $|\Lambda|\rightarrow \infty$ and assuming that $\frac{1}{|\Lambda|}\sum_p \rightarrow (2\pi)^{-3}\int dp$ we obtain the desired variational expression for the Hamiltonian part of the free energy density.

\subsection{The entropy part} We now derive the 
formula for the entropy density in a Bogoliubov trial state in terms
of $\gamma, \alpha$, and $\phi$. To do this we will use some basic facts concerning \textit{Bogoliubov transformations} see,  e.g.\ \cite{NapNamSol-15}.

Given a state $\omega$ with a corresponding density matrix 
${\mathbb G}$, its entropy is defined as
\[
S(\omega)=-\text{Tr}({\mathbb G}\ln {\mathbb G}). 
\]
We only consider Bogoliubov variational states $\omega_{\gamma,\alpha, \phi}$, thus by definitions
(\ref{state}) and (\ref{bogvarstate}) ${\mathbb G}=\mathbb{U}_{\phi}G\mathbb{U}_{\phi}^*$ where
$G$ is the density matrix
corresponding to the quasi-free state $\omega_{\gamma,\alpha}$. Since
$\mathbb{U}_{\phi}$ is unitary we see that
$$\text{Tr}({\mathbb G}\ln {\mathbb G})=\text{Tr}({G}\ln {G})$$
and so
\[
S(\omega_{\gamma,\alpha, \phi})=S(\omega_{\gamma,\alpha}).
\]
This means that the coherent transformation, i.e.\ the condensate, does not change the entropy. 
Thus, if we want to calculate the entropy of  
Bogoliubov trial states it is enough to consider quasi-free states.
For such a state the density matrix $G$ is unitarily equivalent through a Bogoliubov
transformation to an operator of the form
$$
\tilde{G}=Z^{-1}\Pi \exp \left[-\sum_{i\in I}e_{i}a^*_{i}a_{i}\right] 
\Pi,\quad Z=\prod_{i\in I}\frac{1}{1-e^{-e_i}}
$$
where $a_{i}:=a(u_{i})$
for an orthonormal basis $\{u_{i}\}$ of the Hilbert space, $I\subseteq{\mathbb N}$, 
$e_{i} \geq0$, and $\Pi$ is the projection onto the subspace
$\ker\left[\sum_{i\notin I}a^*_{i}a_{i}\right]$. The constant $Z$ (which will be finite) 
ensures that $\text{Tr}(\tilde{G})=1$. 
The 1-pdm 
$\tilde\Gamma$ of $\tilde{G}$ is easily seen to have
$\tilde{\alpha}=0$ and $\tilde\gamma$ diagonal in the basis 
$\{u_{i}\}$ with eigenvalues $\lambda_i$ given by 
$$(1-\exp(-e_{i}))^{-1}=1+\lambda_{i}, \quad i\in I,$$
and zero otherwise. 

For the state above one can easily calculate the entropy. The Fock
space $\Gamma_{\rm{s}}(\cH)$ has the orthonormal basis
$$|\vec{n}\rangle:=|n_{1},n_{2},\ldots\rangle=(n_{1}!n_{2}!\ldots)^{-\frac{1}{2}}(a^*_{1})^{n_{1}}(a^*_{2})^{n_{2}}\ldots|0\rangle,$$
where $|0\rangle$ is the Fock vacuum and $n_{1},n_{2},\ldots\in\mathbb{N}\cup\{0\}$ with only a finite number of $n_{j}$'s that are positive. We find 
\begin{eqnarray*}
\begin{aligned}
S(\tilde{G})&=\sum_{i\in I}\ln(1+\lambda_{i})+\sum_{j\in I}\sum_{\{\vec{n}\}}\langle\vec{n}|\frac{e_{j}a^*_{j}a_{j}\exp[-e_{j}a^*_{j}a_{j}]}{\prod_{i\in I}(1+\lambda_{i})}|\vec{n}\rangle\prod_{i\in I, i\neq j}(1+\lambda_{i}),
\end{aligned}
\end{eqnarray*}
which together with
\begin{eqnarray*}
\begin{aligned}    
&\sum_{\{\vec{n}\}}\langle\vec{n}|e_{j}a^*_{j}a_{j}\exp[-e_{j}a^*_{j}a_{j}]|\vec{n}\rangle =\sum_{n_{j}=0}^{\infty}e_{j}n_{j}e^{-e_{j}n_{j}}=\frac{e_{j}e^{-e_{j}}}{(1-e^{-e_{j}})^{2}}
\end{aligned}
\end{eqnarray*}    
and the definition of $e_{j}$ implies that
\[
\begin{aligned}
S(\tilde{G})&=\sum_{i\in I}\ln(1+\lambda_{i})-\sum_{j\in I}\lambda_{j}\ln\left(\frac{\lambda_{j}}{1+\lambda_{j}}\right)\\ &=\sum_{j\in I}\left[(1+\lambda_{j})\ln(1+\lambda_{j})-\lambda_{j}\ln\lambda_{j}\right].      
\end{aligned}
\]
It is, however, not immediately possible to find the
entropy of $G$ in terms of its 1-pdm $\Gamma$ from this formula.  In fact, although $G$
and $\tilde{G}$ are unitarily equivalent, this is not so for $\Gamma$
and $\tilde\Gamma$.  The relation however is (see \cite{NapNamSol-15})
that
$\Gamma'=(\Gamma+\frac{1}{2}\mathcal{S})^{1/2}\mathcal{S}(\Gamma+\frac{1}{2}\mathcal{S})^{1/2}$
and
$\tilde{\Gamma}'=(\tilde\Gamma+\frac{1}{2}\mathcal{S})^{1/2}\mathcal{S}
(\tilde\Gamma+\frac{1}{2}\mathcal{S})^{1/2}$
are unitarily equivalent.
Since we can express the entropy of $\tilde{G}$ as 
$$
S(\tilde{G})=-\Tr\left((\tilde\Gamma'-\frac{1}{2})\ln|\tilde\Gamma'-\frac{1}{2}|\right).
$$
We have proved the following result. 
\begin{theorem}\label{entropythm}
Let $\omega_{\gamma,\alpha}$ be a quasi-free state with 1-pdm $\Gamma$. The entropy of this state is given by
\[
  S(\omega_{\gamma,\alpha})=-\Tr\left((\Gamma'-\frac{1}{2})\ln|\Gamma'-\frac{1}{2}|\right) 
\]
where $\Gamma'=(\Gamma+\frac{1}{2}\mathcal{S})^{1/2}\mathcal{S}(\Gamma+\frac{1}{2}\mathcal{S})^{1/2}$.
\end{theorem}       
In our case 
$$\Gamma+\frac{1}{2}\mathcal{S}=\begin{pmatrix}
  \gamma+\frac{1}{2} & \alpha\\
  \alpha & \gamma +\frac{1}{2}
 \end{pmatrix}.$$
To calculate the eigenvalues of 
$\Gamma'$ we again use 
the translation invariance of our system and pass to the Fourier space. 
In the momentum representation the eigenvalues are given by
$$\beta(p)=\pm\sqrt{(\frac{1}{2}+\gamma(p))^{2}-\alpha(p)^{2}}$$ 
(note that the eigenvalues of $\Gamma'$ are the same by a similarity transformation as the eigenvalues 
of $\Gamma\mathcal{S}+\frac{1}{2}$)
and we arrive at the desired formula. 
Note that all terms are well-defined since the condition $\Gamma\geq0$ implies that 
\begin{equation*}
\gamma(p)\geq 0 \qquad \text{and} \qquad \gamma(p)(1+\gamma(p))-\alpha(p)^{2}\geq 0.
\end{equation*}

\begin{acknowledgements}
We thank Robert Seiringer and Daniel Ueltschi for bringing the issue of the change in critical temperature to our attention. We also thank the Erwin Schr\"odinger Institute (all authors) and the Department of Mathematics, University of Copenhagen (MN) for the hospitality during the period this work was carried out. We gratefully acknowledge the financial support by the European Union’s Seventh Framework Programme under the ERC Grant Agreement Nos. 321029 (JPS and RR) and 337603 (RR) as well as support by the VILLUM FONDEN via the QMATH Centre of Excellence (Grant No. 10059) (JPS and RR), by the National Science Center (NCN) under grant No. 2012/07/N/ST1/03185 and the Austrian Science Fund (FWF) through project Nr. P 27533-N27 (MN). The authors declare that they have no conflict of interest.
\end{acknowledgements}

\end{document}